\documentclass[11pt]{article}
\usepackage{fullpage}
\usepackage{hyperref}
\usepackage{amssymb}
\usepackage{amsmath}
\usepackage{amsthm}
\usepackage{youngtab}
\usepackage{graphicx,color,colordvi}
\usepackage{subfig}
\usepackage{enumerate}
\usepackage{bbm}
\usepackage[utf8]{inputenc}
\usepackage{fullpage}
\usepackage[blocks]{authblk}
\usepackage{mathtools}
\usepackage{multirow}
\usepackage{dsfont}
\usepackage{appendix}

\newcommand{\nc}{\newcommand}
\nc{\rnc}{\renewcommand}
\nc\mnb[1]{\medskip\noindent{\bf #1}}

\newcommand{\M}{{\mathbb{M}}}

\newcommand{\ot}{\otimes}
\newcommand{\opV}{\operatorname{V}}

\newcommand{\<}{\langle}

\renewcommand{\>}{\rangle}
\newcommand{\A}{\mathcal{A}_{n}^{t_{n}}(d)}

\newcommand{\Res}{\operatorname{Res}}
\newcommand{\Ind}{\operatorname{Ind}}
\newcommand\be{\begin{equation}}
\newcommand\ee{\end{equation}}

\DeclareMathOperator{\tr}{Tr}

\hypersetup{colorlinks, citecolor=magenta, filecolor=blue, linkcolor=blue, urlcolor=green}
\newtheorem{theorem}{Theorem}

\newtheorem{corollary}[theorem]{Corollary}

\newtheorem{definition}[theorem]{Definition}
\newtheorem{example}[theorem]{Example}

\newtheorem{lemma}[theorem]{Lemma}

\newtheorem{proposition}[theorem]{Proposition}
\newtheorem{remark}[theorem]{Remark}

\newtheorem{fact}[theorem]{Fact}
\begin{document}
\title{Optimal Port-based Teleportation}

\author{Marek Mozrzymas}
\affil[1]{\small Institute for Theoretical Physics, University of Wrocław
	50-204 Wrocław, Poland}
\author{Micha{\l} Studzi{\'n}ski}
\author{Sergii Strelchuk}
\affil[1]{\small DAMTP, Centre for Mathematical Sciences, University of Cambridge, Cambridge~CB30WA, UK} 
\author{Micha{\l} Horodecki}
\affil[3]{\small Institute of Theoretical Physics and Astrophysics, National Quantum Information Centre, Faculty of Mathematics, Physics and Informatics, University of Gda{\'n}sk, Wita Stwosza 57, 80-308 Gda{\'n}sk, Poland}
\date{}
\maketitle			 
\begin{abstract}
	Deterministic port-based teleportation (dPBT) protocol is a scheme where a quantum state is guaranteed to be transferred to another system without unitary correction. We characterize the best achievable performance of the dPBT when both the resource state and the measurement is optimized. Surprisingly, the best possible fidelity for an arbitrary number of ports and dimension of the teleported state is given by the largest eigenvalue of a particular matrix -- Teleportation Matrix. It encodes the relationship between a certain set of Young diagrams and emerges as the the optimal solution to the relevant semidefinite program.
\end{abstract}

\section{Introduction}
Quantum teleportation is one of the earliest and most widely used primitives in Quantum Information Science which performs an arbitrary quantum state transfer between two spatially separated systems~\cite{bennett_teleporting_1993}. It involves pre-sharing an entangled resource state and consists of three simple stages. The first stage involves a joint measurement of the teleported subsystem together with the share of the resource state on the sender's side. In the second step, classical measurement outcome is communicated to the receiver. The last step consists of applying a requisite correction operation which recovers the transmitted quantum state.

Port-based teleportation (PBT) discovered by Ishizaka and Hiroshima~\cite{ishizaka_asymptotic_2008} is a particular teleportation protocol which stands out for its simplicity and surprising qualities which are unattainable by the preexisting set of protocols. They were able to reduce the three-step procedure to the one where the remaining correction step is trivial. In this protocol, the sender and the receiver share a large entangled resource state and the sender implements a joint POVM on the teleported system and the resource state. Depending on the type of POVM, one distinguishes two operational regimes: probabilistic and deterministic. 
In the former case, which is well-understood only when one teleports qubits, the measurement is designed to ensure that the teleported state arrives intact to the receiver, but there is a small probability of failure. In the latter case, the state always gets to the receiver but incurs some distortion. 
In both protocols the sender communicates the classical measurement outcome (including the failure in the former case) to the receiver who then traces out part of the resource state indicated by the classical communication and finishing with the teleported state in the case of dPBT or maximally mixed state in case of the probabilistic PBT. 

While the optimal functioning of the probabilistic PBT is well-understood, for a number of practical applications it may be critical to have a teleportation protocol without a unitary correction which always succeeds even when the replica is distorted. Understanding the feasibility of such protocols (with optimal measurements and the corresponding resource state) for an arbitrary number of ports and local dimension of the teleported state remained a difficult open problem. 

Despite the superficial similarity to the probabilistic PBT, characterizing optimal performance of the dPBT remained elusive due to the distortion which affected the teleported state -- the existing tools were ill-suited for the analysis of the resulting quantum state on the receiver. In our work, we show that the optimal performance regime for the dPBT, remarkably, can be reduced to the study of a static object -- Teleportation Matrix. This extraordinarily simple matrix emerges as a result of an SDP optimization, and characterizes the abstract relationship between the input and the output states of the protocol.

In this work we obtain a relationship between the dPBT and its companion Teleportation Matrix and provide a convergent algorithm to determine its infinity norm that characterizes the best possible fidelity of teleportation when both the resource state and measurement are optimized. In particular, when the dimension of the teleported state is greater or equal to the number of ports, the maximal eigenvalue is obtained analytically. In the other case we provide a convergent algorithm to compute it.

In Section~\ref{sec:algebra} we review the connection of PBT protocols with the algebra of partially transposed permutation operators, followed by a short review of basic facts about the induced and restricted representations of the symmetric group $S(N)$ in Section~\ref{sec:sn}. In the same section we also prove a group-theoretic lemma about characters of the induced representations which will play an important role in the following sections. Then, in the first part of Section~\ref{central} we formally introduce the Teleportation Matrix (TM) and study its properties. In particular, we present an analytical expression for its eigenvalues and corresponding eigenvector when the dimension of underlying local Hilbert space is large enough compared to the number of ports. In the second part, we provide an alternative approach to computing spectral properties of the TM. Finally, in Section~\ref{sec:sdp} we show how it naturally appears as a result of semidefinite optimization and describe a convergent algorithm which calculates its infinity norm with corresponding eigenvector when dimension of the local Hilbert space is smaller than number of ports.

\section{The dPBT and its connection to a representation of the algebra}\label{sec:algebra}
We now recall the details of the dPBT introduced in~\cite{ishizaka_remarks_2015,ishizaka_asymptotic_2008,ishizaka_quantum_2009}, and introduce the notation emphasize the connection with the algebra of partially transposed permutation operators $\A$. Here we review the most important facts regarding the representation of $\A$ (for detailed discussion of properties of $\A$ see~\cite{Moz1,Stu2017,Stu1}). In the dPBT, two parties, Alice and Bob, share a resource state consisting of $N$ copies of bipartite maximally entangled states $|\psi^+\>$. Then Alice performs a joint measurement on her half of the resource state and the unknown state $\theta_C$  which she wants to teleport by choosing one of the POVM from the set $\{\widetilde{\Pi}_a\}_{a=1}^N$, where each $\widetilde{\Pi}_a$ is given in the form of square root measurement~\cite{ishizaka_asymptotic_2008}.
She then communicates the measurement outcome $a\in\{1,\ldots,N\}$  to Bob. This outcome $a$ labels the port on Bob's side which contains the teleported state. Bob then traces out all the ports except for the $a$-th. In this protocol, teleportation always succeeds but the teleported state arrives distorted. To characterize the performance of the dPBT we need to evaluate the fidelity of teleportation $F$~\cite{ishizaka_asymptotic_2008}:
\be
\label{fid_def}
F=\frac{1}{d^2}\sum_{a=1}^N\tr\left[\sigma_a\widetilde{\Pi}_a \right] =\frac{1}{d^2}\sum_{a=1}^N\tr\left[\sigma_a\rho^{-1/2}\sigma_a\rho^{-1/2} \right],\quad \widetilde{\Pi}_a=\rho^{-1/2}\sigma_a\rho^{-1/2},
\ee
which is a function of a number of ports $N$ and local dimension of the Hilbert space $d$. For $1\leq a\leq N$
\be
\label{sigma0}
\sigma_a=\frac{1}{d^{N}}\mathbf{1}_{\overline{aC}}\ot \widetilde{P}^+_{aC}=\frac{1}{d^N}\mathbf{1}_{\overline{aC}}\ot V^{t_C}(a,C),
\ee
where $\mathbf{1}_{\overline{aC}}$ denotes the identity operator acting on all subsystems except $a$-th $C$-th, $\widetilde{P}^+_{aC}$ denotes an unnromalised projector onto the maximally entangled state $|\Phi^+\>_{aC}=\frac{1}{\sqrt{d}}\sum_{i=1}^d|ii\>_{aC}$ between subsystems $a$ and $C$, where the set $\{|i\>\}_{i=1}^d$ is the standard basis in $\mathbb{C}^d$. In the second equality in~\eqref{fid_def} we use a well-known fact that $\widetilde{P}^+_{aC}=V^{t_C}(a,C)$, where $t_C$ is a partial transposition with respect to subsystem $C$ performed on permutation operator $V(a,C)$ acting between subsystems $a$ and $C$. The operator $\rho$ in~\eqref{fid_def} is called the PBT operator, and can be expressed as (see~\cite{Stu2017}):
\be
\label{rho0}
\rho=\sum_{a=1}^N\sigma_a=\frac{1}{d^N}\sum_{a=1}^{N}\mathbf{1}_{\overline{aC}}\ot V^{t_C}(a,C)=\frac{1}{d^N}\eta.
\ee
Since every element $\mathbf{1}_{\overline{aC}}\ot V(a,C)$ acts as a permutation on the full Hilbert space $(\mathbb{C}^d)^{\ot n}$, where $n=N+1$, we will further denote it by $V(a,C)$. To keep the notation consistent with the earlier works that study $\A$ we label subsystem $C$ by the index $n$, then expressions~\eqref{sigma0},~\eqref{rho0} read
\be
\label{sigma1}
\sigma_a=\frac{1}{d^N}V^{t_n}(a,n),\quad \rho=\sum_{a=1}^N\sigma_a=\frac{1}{d^N}\sum_{a=1}^NV^{t_n}(a,n)=\frac{1}{d^N}\eta.
\ee
From the above identities it follows that $\rho$ is strictly connected with the algebra $\A$ of partially transposed permutation operator where partial transposition $t_n$ is performed with respect to last $n-th$ subsystem. The operator $\rho$ can be regarded as an element of the algebra $\A$. From~\cite{Moz1,Stu1} we know that the full algebra $\A$ splits into direct sum of two left ideals $\mathcal{A}^{t_n}_n(d)=\mathcal{M}\oplus \mathcal{S}$. From~\cite{Stu2017} we also know that the part of the algebra $\A$ containing the ideal $\mathcal{S}$ does not play any role in the description of the dPBT, so we will not discuss it here. In the ideal $\mathcal{M}$ all irreducible representations (irreps) of $\A$ are labelled by the irreps of the symmetric group $S(N-1)$, and they are strictly connected with the irreps of the group $S(N)$ induced by those irreps of $S(N-1)$.

Furthermore, we denote the corresponding projector (including multiplicities) on chosen irrep labelled by $\alpha \vdash N-1$ (symbol $\vdash$ indicates that the diagram $\alpha$ is obtained for $N-1$ boxes) by  $M_{\alpha}$, and its support space by $S(M_{\alpha})$. Further by $P_{\mu}$ we denote the Young projector (including multiplicities) onto irrep of $S(N)$ labelled by $\mu \vdash N$ induced from a given irrep $\alpha$ of $S(N-1)$. It occurs when a Young diagram $\mu \vdash N$ can be obtained from a Young diagram  $\alpha \vdash N-1$ by adding a single box $\Box$ (we denote this by  $\mu \in \alpha$), and when all irreps labelled by $\alpha$ and $\mu$ occur. The latter happens when the heigh of the first column of $\alpha$ and $\mu$ is less of equal to the dimension $d$ of the local Hilbert space (i.e. when $h(\alpha)\leq d, h(\mu)\leq d$).  Define projectors
\be
\label{op_F}
\forall \ \mu\in \alpha \quad F_{\mu}(\alpha)\equiv M_{\alpha}P_{\mu},
\ee
which project onto irreps of $S(N)$ contained in $M_{\alpha}$ labelled by Young diagrams $\mu$ and induced from the irreps of $S(N-1)$ labelled by $\alpha$~\cite{Stu2017}.
Denoting by $P_{\alpha}$ a Young projector onto irrep labelled by $\alpha \vdash N-1$ we get the following representation of $\eta$ from Eqn.~\ref{sigma1}: 
\be
\eta=\sum_{\alpha}\eta(\alpha)=\sum_{\alpha}V(a,N)P_{\alpha}V^{t_n}(N,n)V(a,N).
\ee
The support of every $\eta(\alpha)$ is the space $S(M_{\alpha})$ which is invariant with respect to action of $S(n-1)$, so we see that $F_{\mu}(\alpha)$ are eigenprojectors of $\eta(\alpha)$.
From~\cite{Stu2017} we know that projectors $F_{\mu}(\alpha)$ can be written as:
\be
\label{op_F2}
F_{\mu}(\alpha)=\gamma^{-1}_{\mu}(\alpha)P_{\mu}\eta(\alpha)P_{\mu},
\ee 
where the numbers $\gamma_{\mu}(\alpha)$ are the eigenvalues of the operator $\eta$ from~\eqref{sigma1} given by
\be
\label{gamma}
\gamma_{\mu}(\alpha)=N\frac{m_{\mu}d_{\alpha}}{m_{\alpha}d_{\mu}},
\ee
where $d_{\alpha},d_{\mu}$ are dimensions of the irreps of $S(N-1),S(N)$ labelled by Young diagrams $\alpha \vdash N-1$, $\mu \vdash N$ respectively, and $m_{\alpha},m_{\mu}$ are their multiplicities.

By combining~\eqref{op_F2} and~\eqref{gamma} we see that the PBT operator $\rho$ which is strictly connected with $\eta$ has the following form:
\be
\label{rodec}
\rho=\sum_{\alpha \vdash N-1}\sum_{\mu \in \alpha}\lambda_{\mu}(\alpha)F_{\mu}(\alpha),
\ee
where
\be
\lambda_{\mu}(\alpha)=\frac{1}{d^N}\gamma_{\mu}(\alpha).
\ee

In our previous work~\cite{Stu2017} we give an explicit expression for the fidelity $F$ given in equation~\eqref{fid_def} in terms of $N,d$, the dimensions $d_{\mu}$, and  multiplicities $m_{\mu}$ of irreps of the permutation group $S(N)$ when the resource state is given by as a $N-$fold tensor product of $|\psi^+\>$. In this case we also know that optimal POVMs $\{\widetilde{\Pi}_a\}_{a=1}^N$ are given in the form of square root measurements (see~\eqref{fid_def}). In the qubit case when both the measurement and the resource state are optimized simultaneously it is known that it is possible to achieve a significantly higher teleportation fidelity~\cite{ishizaka_quantum_2009}. In the latter case, the resource state differs from $|\psi^+\>^{\ot N}$, and one has a different set of POVMs. In the qudit case we similarly take the resource state to be
\be
|\Psi\>=\left(O_A\ot \mathbf{1}_B\right)|\psi^+\>_{A_1B_1}\ot |\psi^+\>_{A_2B_2}\ot \cdots \ot |\psi^+\>_{A_NB_N}, 
\ee
where $A=A_1A_2\cdots A_N$, $B=B_1B_2\cdots B_N$, and $\tr O_A^{\dagger}O_A=d^N$, where $O_A$ encodes an arbitrary quantum operation on Alice's side. We want to compute
\be
\label{vv}
F=\frac{1}{d^2}\max_{\{\Pi_a\}}\sum_{a=1}^N\tr\left[\Pi_a\sigma_a \right],
\ee
with respect to the following constraints
\be
\label{bb1}
(1)\quad \sum_{i=a}^N\Pi_a\leq X_A\ot \mathbf{1}_{\overline{B}},\quad (2)\quad \tr X_A=\tr O_A^{\dagger}O_A=d^N,
\ee
where $\{\Pi_a\}_{a=1}^N$ is some new, optimal set of POVMs which are compatible with operation $O_A$ and $\mathbf{1}_{\overline{B}}$ is identity operator acting on single qudit space on Bobs' side. We see that the problem of simultaneous optimisation over a resource state $|\Psi\>$ and the set of POVMs $\{\Pi_a\}_{a=1}^N$ can be cast as a semi-definite program (SDP)~\cite{Boyd}. 
If we are interested in optimizing only the measurement then see~\cite{Stu2017}, and for explicit formula in the case of small number of ports see~\cite{wang_higher-dimensional_2016}.
Most of this work is dedicated to finding an optimal form of the Alice operation $O_A$, optimal form of POVMs, and expression for the optimal value of the fidelity~\eqref{vv}. 
As we have mentioned above we solve this problem by giving an analytical solution of the primal and the dual SDP. Moreover, all such solutions are presented in terms of objects characterising $\A$.
\section{Facts about symmetric group $S(N)$}\label{sec:sn}
\label{Sn}
Before we state and prove our results, we need to introduce further group-theoretic notation.

	\begin{enumerate}[i)]
	\item By the symbol $\nu /\mu=\square$ we denote two Young diagrams $\mu,\nu$ for the same natural number $N$ when $\mu$ can be obtained from $\nu$ by moving a single box $\square$ (and vice versa). 
	\item By $\alpha \in \mu$ we denote Young diagrams $\alpha \vdash N-1$ which can be obtained from $\mu \vdash N$ by removing one box $\Box$.
	\item By $\widehat{S}(N)$ we denote the set of all possible irreps of the symmetric group $S(N)$, and by $|\widehat{S}(N)|$ its cardinality.
	\item By $\varphi^{\alpha},\psi^{\mu}$, etc. we denote irreps of respective symmetric groups belonging to sets $\widehat{S}(N-1)$ or $\widehat{S}(N)$.
	\item For every permutation $\sigma\in S(N)$ we define its decomposition into disjoint cycles $\sigma =(1^{k},2^{\xi _{2}},\ldots,N^{\xi _{N}}),$ where $
	k\geq 1,$ $\xi_{i}\geq 0$, $i=2\ldots N$ denote the number of cycles of the length $1$ to $N$.
	\end{enumerate}

Recall that the representations $\Res_{S(N-1)}^{S(N)}(\psi^{\nu} )$, $\psi^{\nu} \in \widehat{S}(N)$ and $\Ind_{S(N-1)}^{S(N)}(\varphi^{\alpha})$ $\varphi^{\alpha} \in \widehat{S}(N-1)$, have the following structure
\be
\label{res}
\Res_{S(N-1)}^{S(N)}(\psi^{\nu} )=\bigoplus _{\alpha \in\nu}\varphi
^{\alpha },\qquad \Ind_{S(N-1)}^{S(N)}(\varphi^{\alpha} )=\bigoplus _{\mu \in \alpha
}\psi ^{\mu },
\ee
so they are simply reducible.  
The following properties of $\Res_{S(N-1)}^{S(N)}(\psi^{\nu} )$ and $\Ind_{S(N-1)}^{S(N)}(\varphi^{\alpha})$ will be required in Section~\ref{central}:
\begin{proposition}
	\label{p_corr}
	\begin{enumerate}[a)]
		We have the following:
		\item  $\varphi^{\alpha} \in \Res_{S(N-1)}^{S(N)}(\psi^{\nu})$ if and only if $
		\psi^{\nu} \in \Ind_{S(N-1)}^{S(N)}(\varphi^{\alpha} )$.
		
		\item Irreps $\psi^{\mu} ,\psi^{\nu} \in \widehat{S}(N),$ $\mu \neq $ $\nu $ are
		in the relation $\nu/ \mu =\square $ if and only if there exists $
		\varphi^{\alpha} \in \Res_{S(N-1)}^{S(N)}(\psi^{\nu} ):\psi^{\mu} \in
		\Ind_{S(N-1)}^{S(N)}(\varphi^{\alpha} )$.
	\end{enumerate}
\end{proposition}

\begin{proof}
	The statement a) of the Proposition is a well-known result in representation
	theory. We prove part b). From the assumption we have 
	\be
	\nu =(\nu _{1},\ldots,\nu _{k},\ldots,\nu _{l},\ldots,\nu _{p})\Rightarrow \mu =(\nu
	_{1},\ldots,\nu _{k}-1,\ldots,\nu _{l}+1,\ldots,\nu _{p})
	\ee
	for some indices $k,l$. We chose 
	\be
	\alpha =(\nu _{1},\ldots,\nu _{k}-1,\ldots,\nu _{l},\ldots,\nu _{p}) \vdash N-1,
	\ee
	which is properly defined Young diagram because by assumption $\mu $ is
	properly defined Young diagram and we have $\mu \in \alpha$, so $\psi^{\mu}
	\in \Ind_{S(N-1)}^{S(N)}(\varphi^{\alpha})$. On the other hand for  b) we
	have from the assumption that for a given $\nu =(\nu _{1},\ldots,\nu _{s},\ldots,\nu
	_{t},\ldots,\nu _{q})$ such that $s\neq t$  
	\be
	\alpha =(\nu _{1},\ldots,\nu _{s},\ldots,\nu _{t}-1,\ldots,\nu _{p}),\qquad \mu =(\alpha
	_{1},\ldots,\alpha _{s}+1,\ldots,\alpha _{t},\ldots,\alpha _{q}),
	\ee
	so  $\mu =(\nu _{1},\ldots,\nu _{s}+1,\ldots,\nu _{t}-1,\ldots,\nu _{q})$ and $\nu /
	\mu=\square$.
\end{proof}

We further prove the following useful statement about characters of the induced representations. 

\begin{lemma}
	\label{L2}
	Let $\sigma \in S(N)$ and suppose that $\sigma $ has the following cycle
	structure $\sigma =(1^{k},2^{\xi _{2}},\ldots,N^{\xi _{n}})$, then 
	\be
	\chi ^{\operatorname{Ind}_{S(N-1)}^{S(N)}(\varphi^{\alpha} )}(\sigma )=k\chi ^{\alpha
	}(1^{k-1},2^{\xi _{2}},\ldots,N^{\xi _{n-k}}). 
	\ee
	In particular  for $\sigma=\operatorname{e}\in (1^{N})$, where $\operatorname{e}$ denotes identity element of the group $S(N)$ we have
	\be
	\chi ^{\operatorname{Ind}_{S(N-1)}^{S(N)}(\varphi^{\alpha} )}(\operatorname{e})=Nd_{\alpha }. 
	\ee
\end{lemma}

\begin{proof}
	Recall that the induced representation $%
	\Ind_{S(N-1)}^{S(N)}(\varphi^{\alpha} ):$ $\varphi^{\alpha} \in \widehat{S}(N-1)$ has the
	following form 
	\be
	\forall \sigma \in S(N)\qquad \Phi _{ai,bj}^{\Ind(\varphi^{\alpha} )}(\sigma )=
	\widetilde{\varphi }_{ij}^{\alpha }[(aN)\sigma (bN)],
	\ee
	where 
	\be
	\widetilde{\varphi }_{ij}^{\alpha }(\pi )=
	\begin{cases}
		\varphi ^{\alpha }(\pi ), \ \pi \in S(N-1), \\ 
		0, \ \pi \notin S(N-1),
	\end{cases}
	\ee
	and $a,b=1,\ldots,N$. We thus get the following formula for the character
	of the induced representation 
	\be
	\chi ^{\operatorname{Ind}_{S(N-1)}^{S(N)}(\varphi^{\alpha} )}(\sigma )=\sum_{i=1}^{d_{\alpha
	}}\sum_{a=1}^{N}\widetilde{\varphi }_{ii}^{\alpha }[(aN)\sigma
	(aN)]=\sum_{a=1}^{N}\widetilde{\chi }^{\alpha }[(aN)\sigma (aN)],
	\ee
	where $\widetilde{\chi }^{\alpha }$ is defined in the same way as $
	\widetilde{\varphi }_{ij}^{\alpha }$. Let $\sigma =C_{1}C_{2}\cdots C_{k}\in S(N)
	$ be a unique decomposition of the permutation $\sigma $ into disjoint
	cycles. For a given transposition $(aN)$ of the natural transversal, the
	number $a$ appears in only one cycle $C_{i}$ in $\sigma $, and similarly for
	the number  $N$ and we have the following possible cycles, which
	include the numbers $a,N$ 
	\be
	\label{possibilities}
	\begin{split}
		&(aN)(ai_{1}\cdots i_{p})(aN)=(Ni_{1}\cdots i_{p}),\quad i_{k}\neq N,\\
		&(aN)(Ni_{1}\cdots i_{p})(aN)=(ai_{1}\cdots i_{p}),\quad i_{k}\neq a,\\	
		&(aN)(ai_{1}\cdots N\cdots i_{p})(aN)=(Ni_{1}\cdots a\cdots i_{p}),\quad i_{k}\neq a,N.
	\end{split}
	\ee
	From Eqn.~\eqref{possibilities} it follows that if $\sigma =C_{1}C_{2}\cdots C_{k}\in S(N)$ is such
	that $|C_{i}|>1$ (i.e. all cycles $C_{i}$ in $\sigma $ are of the length
	greater than one), then for any  transposition $(aN)$ the permutation $
	(aN)\sigma (aN)$ does not belong to $S(N-1)$, and $\chi^{\operatorname{Ind}_{S(N-1)}^{S(N)}}(\sigma )=\sum_{a=1}^{N}\widetilde{\chi }^{\alpha }[(aN)\sigma (aN)]=0$.
	Suppose now that a permutation $\sigma $ contains the cycle of the length
	one i.e. it is of the form 
	\be
	\sigma \in (1^{k},2^{\xi _{2}},\ldots,(N-k)^{\xi _{n-k}}),\quad k\geq
	1, \xi _{j}\geq 0,\qquad \sigma =(a_{1})(a_{2})\cdots (a_{k})C_{1}\cdots C_{p},
	\ee
	where $a_{i}=1,\ldots,N$ and $|C_{j}|>1$. In this case we have for $i=1,\ldots,k$  
	\be
	(a_{i}N)\sigma (a_{i}N)=(a_{1})\cdots(N)\cdots (a_{k})C_{1}'\cdots C_{p}'\in S(N-1),
	\ee
	so for $k$ transpositions of the transversal $(a_{i}N): \ i=1,\ldots,k$ we have 
	\be
	\widetilde{\chi }^{\alpha }[(a_{i}N)\sigma (a_{i}N)]=\chi ^{\alpha
	}(1^{k-1},2^{\xi _{2}},\ldots,(N-k)^{\xi _{n-k}})
	\ee
	and for the remaining transpositions of the transversal $(a_{j}N):$ $j>k$ we
	have 
	\be
	\widetilde{\chi }^{\alpha }[(a_{j}N)\sigma (a_{j}N)]=0,
	\ee
	and 
	\be
	\chi^{\operatorname{Ind}_{S(N-1)}^{S(N)}(\varphi^{\alpha})}(1^{k},2^{\xi _{2}},\ldots,(N-k)^{\xi _{n-k}})=k\chi
	^{\alpha }(1^{k-1},2^{\xi _{2}},\ldots,(N-k)^{\xi_{n-k}}).
	\ee
\end{proof}

\section{Teleportation Matrix}\label{central}
We are now ready to define the central object of our work -- the Teleportation Matrix $M_F$ which plays a key role in the analysis of the simultaneous optimisation over POVMs and the resource state in the dPBT. Later, we will derive a connection between $M_F$ and induced characters of the symmetric group which enables us to use results from Section~\ref{Sn} in order to determine its spectral properties. We provide an analytical expression for its eigenvalues whenever $d\geq N$, and show that $M_F$ together with all of its principal submatrices is positive semi-definite. Finally, we derive a few other important properties of $M_F$ like its irreducibility and primitivity which are necessary when we discuss the convergent algorithm for computation of the infinity norm of principal submatrices of $M_F$ (i.e. when $d<N$ and the closed-form analytical expression for the eigenvalues is not known).
\begin{definition}
\label{M_F}
Let $\mu, \nu$ run over all  irreps  of the group $S(N)$, define the following matrix $M_F$ of dimension $|\widehat{S}(N)|$
\be
\label{MF}
M_F\equiv (n_{\mu}\delta_{\mu, \nu}+\Delta_{\mu, \nu}),
\ee
where $n_{\mu}$ is the number of $\alpha \vdash N-1$ for which $\alpha \in \mu$, and 
\be
\Delta_{\mu,\nu}=\begin{cases}
	1 \ \text{if} \ \mu/\nu=\Box,\\
	0 \ \text{otherwise}.
\end{cases}
\ee
The symbol $\mu/\nu=\Box$ denotes such Young diagrams $\mu,\nu$ which can be obtained from each other by moving a single box. 
\end{definition}
Fig~\ref{fig:teleportationmatrix} depicts $M_F$ for $N=2,3,4$ when all the irreps of $S(N)$ occur.
\begin{figure}[ht!]
	\centering
	\includegraphics[width=100mm]{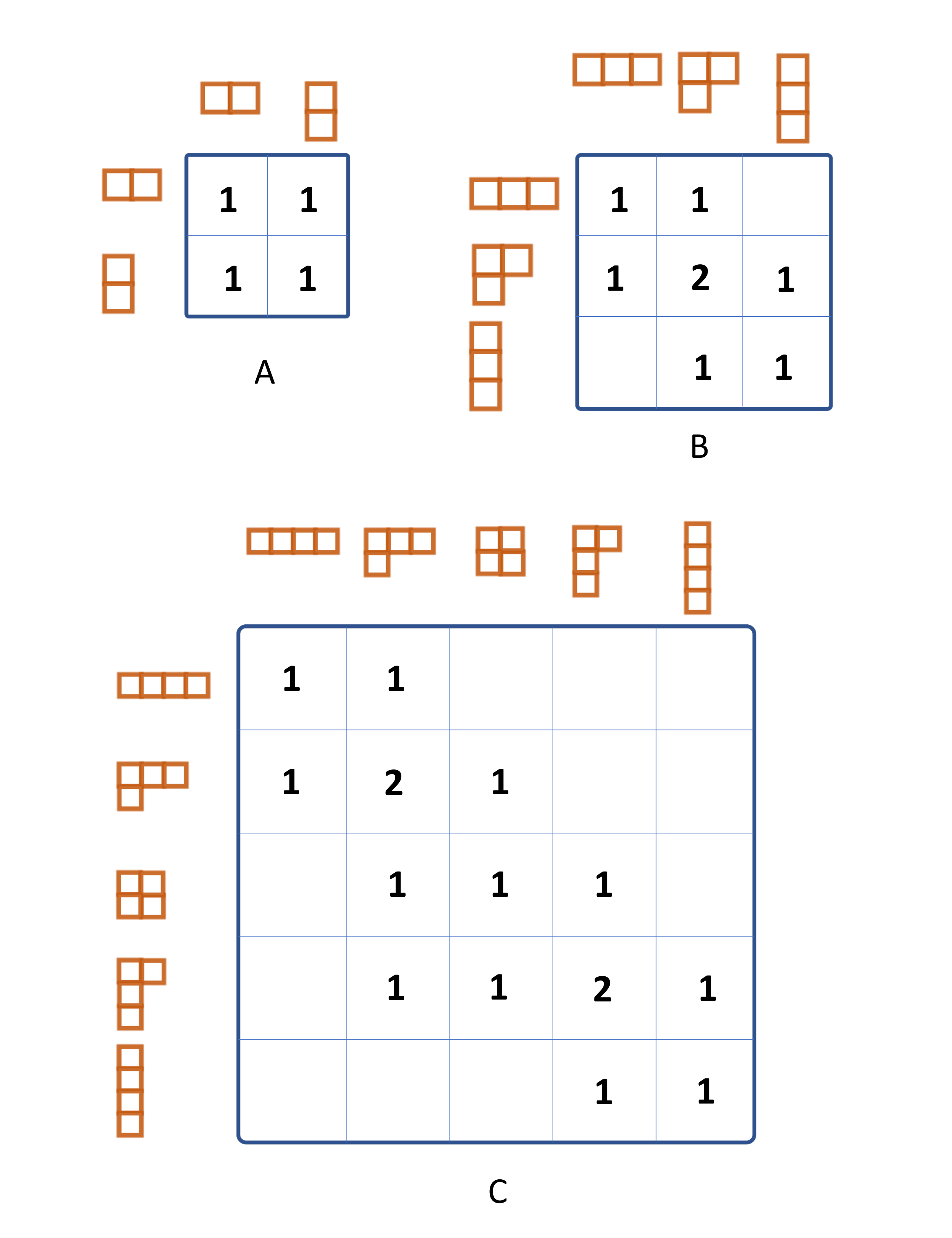}
	\caption{Teleportation matrix for the dPBT schemes. The maximal eigenvalue of each of the matrices determines the optimal performance of the dPBT scheme for: $N=2$ (A), $N=3$ (B), $N=4$ (C) in the case where all the irreps occur (i.e. local dimension $d$ of the teleported state and each of the port equals to $N$). Empty squares are filled with zeros.}
	\label{fig:teleportationmatrix}
\end{figure}
From the representation theory point of view, the structure of $M_F$ encodes relations among the irreps of the group $S(N)$. As we will see later, the relations that define the matrix $M_F$ are determined by the properties of the representations $\operatorname{Res}$ and $\operatorname{Ind}$ (see Section~\ref{Sn}). We will further assume that all indices $\psi^{\mu},\psi^{\nu}\in \widehat{S}(N)$ of the matrix $M_F$ are ordered in the strongly decreasing lexicographic order, starting from the biggest Young diagram $\mu=(N)$. In such ordering, Young diagrams strongly decrease, whereas the height of the Young diagrams weakly increases. \\To reveal the connection between $M_F$ and irreps of $S(N)$ we start from the following lemma:
\begin{lemma}
	\label{L1}
	The numbers, which appear in the row $\nu$ of the matrix 
	$M_{F},$ are the multiplicities of the irreps $\psi^{\nu} \in \widehat{S}(N)
	$ appearing in all representations 
	\be
	\operatorname{Ind}_{S(N-1)}^{S(N)}(\varphi^{\alpha} ):\varphi^{\alpha} \in \operatorname{Res}_{S(N-1)}^{S(N)}(\psi^{\nu}
	),\quad \varphi^{\alpha} \in \widehat{S}(N-1),
	\ee
	where the diagonal term $n_{\nu }$ shows how many $\varphi^{\alpha} \in \operatorname{Res}%
	_{S(N-1)}^{S(N)}(\psi^{\nu} )$.
\end{lemma}

\begin{proof}
	The lemma is in fact, a corollary from the
	Proposition~\ref{p_corr}. From the statement a) of this proposition we get that for a given $
	\psi^{\nu} \in \widehat{S}(N)$, so for a given row $\nu $ of the matrix $M_{F}$,
	the irrep $\nu $ is included in all representations $%
	\Ind_{S(N-1)}^{S(N)}(\varphi^{\alpha} )$ such that $\varphi^{\alpha} \in \Res_{S(N-1)}^{S(N)}(\psi^{\nu} )$, and there are $n_{\nu }$ of them. From statement b)
	of Proposition~\ref{p_corr} we get that if $\mu \neq $ $\nu $ then $\nu /\mu 
	=\square $ if and only if $\psi^{\mu} $ belongs to $\Ind_{S(N-1)}^{S(N)}(\varphi^{\alpha} )$ for some $%
	\varphi^{\alpha} \in \Res_{S(N-1)}^{S(N)}(\psi^{\nu} )$. It is not difficult to prove
	that in the case $\mu \neq $ $\nu $ the irrep $\mu : \nu /\mu 
	=\square $ appears only once in all $\Ind_{S(N-1)}^{S(N)}(\varphi^{\alpha} ):\varphi^{\alpha} \in 
	\Res_{S(N-1)}^{S(N)}(\psi^{\nu} )$.
\end{proof}

From the point of view of representation theory, the structure of $M_F$ encodes relations among the irreps of $S(N)$. Such relations are determined by the properties of the representations $\operatorname{Res}$ and $\operatorname{Ind}$ (see Section~\ref{Sn}). In what follows we assume that all indices $\psi^{\mu},\psi^{\nu}\in \widehat{S}(N)$ of the matrix $M_F$ are in the strongly decreasing, lexicographic order, starting from $\mu=(N)$. In such ordering Young diagrams strongly decrease, whereas their heights weakly increase. \\To reveal the connection between $M_F$ and irreps of $S(N)$ the first prove the following lemma:
\begin{lemma}
	\label{L1}
	The numbers, which appear in the row $\nu$ of the matrix 
	$M_{F},$ are the multiplicities of the irreps $\psi^{\nu} \in \widehat{S}(N)
	$ appearing in all representations 
	\be
	\operatorname{Ind}_{S(N-1)}^{S(N)}(\varphi^{\alpha} ):\varphi^{\alpha} \in \operatorname{Res}_{S(N-1)}^{S(N)}(\psi^{\nu}
	),\quad \varphi^{\alpha} \in \widehat{S}(N-1),
	\ee
	where the diagonal term $n_{\nu }$ shows how many $\varphi^{\alpha} \in \operatorname{Res}%
	_{S(N-1)}^{S(N)}(\psi^{\nu} )$.
\end{lemma}

\begin{proof}
	The lemma is in fact, a corollary from the
	Proposition~\ref{p_corr}. From the statement a) of this proposition we get that for a given $
	\psi^{\nu} \in \widehat{S}(N)$, so for a given row $\nu $ of the matrix $M_{F}$,
	the irrep $\psi^{\nu} $ is included in all representations $%
	\Ind_{S(N-1)}^{S(N)}(\varphi^{\alpha} )$ such that $\varphi^{\alpha} \in \Res_{S(N-1)}^{S(N)}(\psi^{\nu} )$, and there are $n_{\nu }$ of them. From statement b)
	of Proposition~\ref{p_corr} we get that if $\mu \neq $ $\nu $ then $\nu /\mu 
	=\square $ if and only if $\psi^{\mu} $ belongs to $\Ind_{S(N-1)}^{S(N)}(\varphi^{\alpha} )$ for some $%
	\varphi^{\alpha} \in \Res_{S(N-1)}^{S(N)}(\psi^{\nu} )$. It is not difficult to prove
	that in the case $\mu \neq $ $\nu $ the irrep $\mu : \nu /\mu 
	=\square $ appears only once in all $\Ind_{S(N-1)}^{S(N)}(\varphi^{\alpha} ):\varphi^{\alpha} \in 
	\Res_{S(N-1)}^{S(N)}(\psi^{\nu} )$.
\end{proof}

In order to describe the spectral properties of the matrix $M_{F}$ we introduce a notion of reduced character
\begin{definition}
	\label{reducedT}
	The reduced character matrix for the group $S(N)$ has the following form%
	\be
	T\equiv(\chi _{\mu }(C)), 
	\ee
	where $\mu $ runs over all irreps of the group $S(N)$, $%
	C=(1^{k},2^{\xi_{2}},\ldots,N^{\xi _{N}})$ describes the class of
	conjugated elements,  $\chi _{\mu }(\cdot)$ is character calculated on irrep $\mu$ and elements from $C$. By $T(C)=(\chi
	_{\mu }(C)),$ where $C$ runs over all classes of the group $S(N)$, we denote the columns of the matrix $T$. 
\end{definition}
Matrix $T=(\chi _{\mu }(C))$ is unitary and related to $M_{F}$ via:
\begin{proposition}
	\label{spec1}
	We have the following spectral properties of the matrix $M_{F}$%
	\be
	M_{F}T(C)=kT(C)\Leftrightarrow \sum_{\mu }(M_{F})_{\nu \mu }\chi _{\mu
	}(C)=k\chi _{\nu }(C), 
	\ee
	where $C=(1^{k},2^{\xi _{2}},\ldots,N^{\xi _{N}})$, so $k$ is the number of
	cycles of the length $1$ in the class $C$ which is the support of the
	eigenvector $T(C)$. The reduced character matrix $T$ for the
	group $S(N)$, diagonalises the matrix $M_{F}$.
\end{proposition}

\begin{proof}
	From Lemma~\ref{L1} we deduce that for the given row $\nu $ of the
	matrix $M_{F}$ the sum 
	\be
	\sum_{\mu }(M_{F})_{\nu \mu }\chi _{\mu }(C)
	\ee
	is equal to the sum of all characters of the irreps of the group $
	S(N)$ which are included in all induced representations $
	\Ind_{S(N-1)}^{S(N)}(\varphi^{\alpha} ):\varphi^{\alpha} \in \Res_{S(N-1)}^{S(N)}(\psi^{\nu}
	), \varphi^{\alpha} \in \widehat{S}(N-1)$ i.e. we have 
	\be
	\sum_{\mu }(M_{F})_{\nu \mu }\chi _{\mu }(C)=\sum_{\varphi^{\alpha} \in \Res_{S(N-1)}^{S(N)}(\psi^{\nu} )}\chi ^{\Ind_{S(N-1)}^{S(N)}(\varphi^{\alpha} )}(C),
	\ee
	where $C=(1^{k},2^{\xi _{2}},\ldots,(N-k)^{\xi_{n-k}})$. From Lemma~\ref{L2} we
	have 
	\be
	\sum_{\varphi^{\alpha} \in \Res_{S(N-1)}^{S(N)}(\psi^{\nu} )}\chi
	^{\Ind_{S(N-1)}^{S(N)}(\varphi^{\alpha} )}(C)=k\sum_{\varphi^{\alpha} \in \Res_{S(N-1)}^{S(N)}(\psi^{\nu} )}\chi ^{\alpha }(1^{k-1},2^{\xi
		_{2}},\ldots,(N-k)^{\xi _{n-k}}),
	\ee
	where the sum on $RHS$ is the character of the representation $\Res_{S(N-1)}^{S(N)}(\psi^{\nu} )$, and we have 
	\be
	\sum_{\varphi^{\alpha} \in \Res_{S(N-1)}^{S(N)}(\psi^{\nu} )}\chi ^{\alpha
	}(1^{k-1},2^{\xi _{2}},\ldots,(N-k)^{\xi _{n-k}})=\chi
	_{v}(1^{k},2^{\xi _{2}},\ldots,(N-k)^{\xi _{n-k}})=\chi _{\nu }(C).
	\ee
\end{proof}

From Proposition~\ref{spec1} one can get:
\begin{corollary}
	\label{c8}
	\begin{enumerate}
	\item The matrix $M_{F}$ has the following spectrum%
	\be
	\operatorname{spec}(M_F)=\{0,1,2,\ldots,N-2,N\}. 
	\ee
	Note that there is a gap in this spectrum -- the number $N-1$ does not occur.
	
	\item The matrix $M_{F}$ is positive semi-definite.
	
	\item The multiplicity of the eigenvalue\ $k\in \operatorname{spec}(M_F)$ is equal to
	the number of cycles classes of the form $(1^{k},2^{\xi
		_{2}},\ldots,N^{\xi_{n}})$, equivalently to the number of solutions in $\mathbb{N}\cup \{0\}$ of the equation (equations for $\xi _{l})$ 
	\be
	\sum_{l=2}^{N-k}l_{\xi _{l}}=N-k.
	\ee
	\item The eigenvector $v=(v_{\mu})$ for $\mu \in \widehat{S}(N)$ corresponding to maximal eigenvalue $N$ has strictly positive entries (which agrees with Frobenius-Perron Theorem - see Theorem~\ref{Perron} of Appendix~\ref{C}) and $\forall \mu \in \widehat{S}(N) \ v_{\mu}=d_{\mu}$, where $d_{\mu}$ is the dimension of the respective irrep.
	\item The largest eigenvalue $N$, in fact spectral radius, has multiplicity
	one, which agrees with Frobenius-Perron Theorem. Similarly the eigenvalues $
	N-2,$ $N-3$ also are simple and the multiplicities of the eigenvalues $N-4,$ $
	N-5$ are equal $2$ and so on.
	\end{enumerate}
\end{corollary}
The above statements are true when all irreps of $S(N)$ occur. This happens whenever heights $h(\mu),h(\nu)$ of Young diagrams labelling rows and columns of $M_F$ satisfy conditions $h(\mu)\leq d,h(\nu)\leq d$.  The minimal dimension $d$ for having all irreps is just equal to the heigh of the Young diagram corresponding to antisymmetric space, so it occurs when $d\geq N$. 

To make our exposition more transparent, we introduce the following
\begin{definition}
	\label{NR1}
	If $\psi^{\mu} \in \widehat{S}(N)$ is irrep of the group $S(N)$ we write
	\be
	\widehat{S}_{d}(N)=\{\psi^{\mu} \in \widehat{S}(N) :  h(\mu )\leq d\}\Rightarrow 
	\widehat{S}_{N}(N)=\widehat{S}(N).
	\ee
\end{definition}
Thus whenever $d$ is small that the height of a for Young diagrams spectral analysis reduces to that of the respective principal submatrices of $M_F$ defined as follows

\begin{definition}
	By $M_{F}^{d}$ we denote a principal submatrix (i.e. matrix
	localised on the main diagonal in the upper left corner), which contains all irreps $\psi^{\nu} \in \widehat{S}(N)$, such that $h(\nu )\leq d$. For such choice we have 
	\be
	N\leq d\Rightarrow M_{F}^{d}=M_{F}, 
	\ee
	and in particular $M_{F}^{N}=M_{F}$.
\end{definition}

Fig~\ref{fig:teleportationmatrix2} illustrates $M_F$ with its principal submatrices $M_F^d$ for $N=5$ when $d=2,3,4,5$.
\begin{figure}[ht!]
	\centering
	\includegraphics[width=110mm]{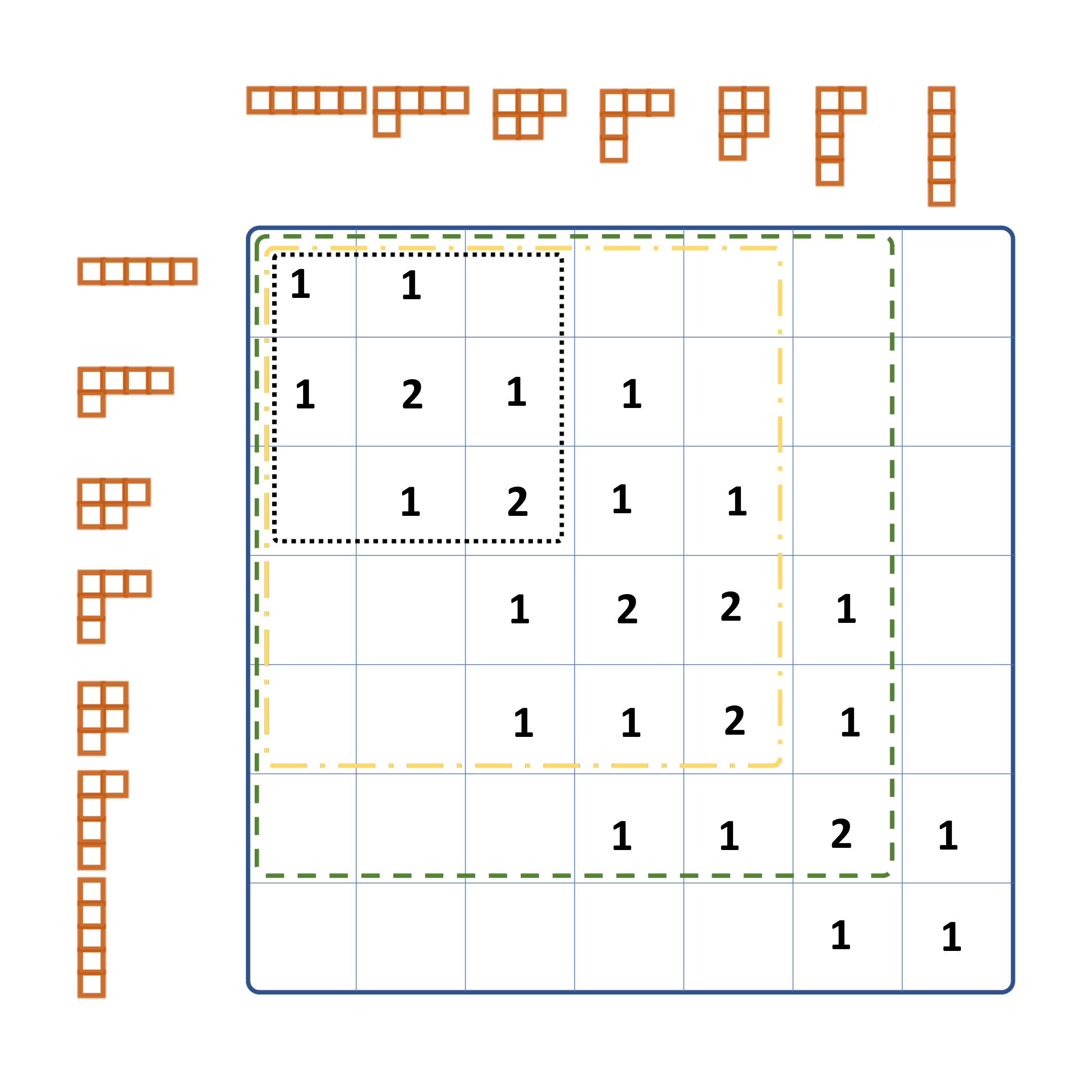}
	\caption{Teleportation matrix for the dPBT schemes with fixed number of ports ($N=5$ -- number of boxes in each shape) and varying dimensions of each port and teleported state (the maximum admissible height of each shape). A sequence of principal submatrices corresponds to an optimal performance of a different dPBT scheme: the entire matrix (solid blue frame) corresponds to $d\geq5$, and its first principal submatrix (dashed green frame) corresponds to the dPBT $d=4$, followed by $d=3,2$ (dash dotted yellow frame and dotted black frame respectively). Empty cells contain zeros.}
	\label{fig:teleportationmatrix2}
\end{figure}

\begin{remark}
\label{posMd}
From Sylvester's theorem (see Theorem~\ref{Sylw} of Appendix~\ref{C}) it follows that all principal matrices $M_{F}^{d}$ are positive semi-definite.
\end{remark}

Using Lemma~\ref{L1} we can calculate how many irreps $\psi^\nu$ of $S(N)$ we have in $\operatorname{Ind}_{S(N-1)}^{S(N)}(\varphi^{\alpha}
):\varphi^{\alpha} \in \operatorname{Res}_{S(N-1)}^{S(N)}(\psi^{\nu} )$ (i.e. how many $1^{\prime }s$ (with multiplicities) we have in the row $\nu $ in the matrix $M_{F}^{d}$):
\begin{proposition}
	\label{above}
	The number of all $\operatorname{Ind}_{S(N-1)}^{S(N)}(\varphi^{\alpha} ):\varphi^{\alpha} \in \operatorname{Res}_{S(N-1)}^{S(N)}(\psi^{\nu} )$ is not greater than $h(\nu )\leq d$, so $n_{\nu
	}\leq d.$ In each induced representation $\operatorname{Ind}_{S(N-1)}^{S(N)}(\varphi^{\alpha} )$ we
	have at most $h(\nu )+1$  irreps of $ S(N)$, if $h(\nu )<d$, 
	and $d$  irreps of $
	S(N)$ if $h(\nu )=d$. From this it follows that in the matrix $M_{F}^{d}$, the maximum
	number of $1^{\prime }s$ (with multiplicities) in each row is not greater than $d^{2}$.
\end{proposition}
Defining $||A||_{1}\equiv \max_{i}\sum_{j=1}|a_{ij}|,$ for an arbitrary  $A=(a_{ij})\in \M(n,\mathbb{C})$, and using  Proposition~\ref{above} we have the following

\begin{corollary}
	\label{cor8}
	We have the following upper bound for the norm of the matrix  norm of  $%
	M_{F}^{d}$ 
	\be
	||M_{F}^{d}||_{1}\leq d^{2}.
	\ee
\end{corollary}

We now exhibit a few additional important features of the teleportation matrix $M_F$, and its principal matrices $M_{F}^{d}$.
It turns out that matrices $M_{F}^{d}$ have a few 
useful properties regarding our algorithm presented further in  Section~\ref{algo}--  irreducibility and primitivity which are explained in Definition~\ref{irr_mat0}, Definition~\ref{irr_mat}, and Definition~\ref{primi} of Appendix~\ref{C}.
\begin{fact}
	\label{M_F_irr}
	$M_F$ given in Definition~\ref{M_F} is irreducible in the sense of Definition~\ref{irr_mat}.
\end{fact}
\begin{proof}
	From the Definition~\ref{M_F} we see that the matrix $M_F$ is at least three-diagonal. The number of zeros in every row of the matrix $M_F$ is equal then to $m=|\widehat{S}(N)|-2$. After the exponentiation of $M_F^2$ the positions $(M_F)_{1,3}\neq 0,\ldots,(M_F)_{|\widehat{S}(N)|-2,|\widehat{S}(N)|}\neq 0$, so the third upper (lower) diagonal becomes nonzero. Computing $M_F^3$ we see that the fourth upper (lower) diagonal has strictly positive entries. Because of the construction continuing process of the multiplication $m+1=|\widehat{S}(N)|-1$ times we have $(A^{m+1})_{ij}>0$ for every $1\leq i,j\leq |\widehat{S}(N)|$. In general matrix $M_F$ has strictly positive numbers also outside of the three main diagonals. It means that in the general case the required number of the multiplications can be smaller than $m+1$.  
\end{proof}
Using similar arguments as in Fact~\ref{M_F_irr} we can show that every principal matrix $M_F^d$ is also irreducible. Matrix $M_F$, and its principal matrices $M_F^d$ are also primitive matrices (see Definition~\ref{primi} of Appendix~\ref{C}). Matrices $M_F$, $M_F^d$ satisfy all the assumptions of Proposition~\ref{primi_prop} of Appendix~\ref{C} so we get:

\begin{corollary}
	\label{c_prim}
	The matrices $M_F, M_{F}^{d}$ are primitive.
\end{corollary}

\begin{remark}
	It follows also directly from the positive semi-definiteness of the matrices $%
	M_{F}^{d}$.
\end{remark}
And lastly
\begin{remark}
The matrix $M_F$ given in the Definition~\ref{MF} is a centrosymmetric matrix according to Definition~\ref{centro} of Appendix~\ref{C}.
\end{remark}

\subsection{A different approach to eigenvalue analysis of the Teleportation Matrix}
\label{sub2a}
We will now exhibit an entirely different approach to finding spectrum of $M_F$. Recall Definition~\ref{NR1} and define the following matrix:
\begin{definition}
	For every $N\geq 2$ we define
	\label{def_of_R}
	\be
	R_{N}^{d}\equiv (r_{\alpha \mu }^{d}(N))\in \mathbb{M}(\widehat{S}_{d}(N-1)\times \widehat{
		S}_{d}(N),\mathbb{Z}),
	\ee
	where 
	\be
	r_{\alpha \mu }^{d}(N)=%
	\begin{cases}
		1:\mu \in \alpha,  \\ 
		0:\mu \notin \alpha.
	\end{cases}
	\ee
	The matrix $R_{N}^{d}$ has its rows indexed by irreps $\varphi^{\alpha} \in \widehat{S}_{d}(N-1)$ whereas the columns are
	indexed by irreps $\psi^{\mu} \in \widehat{S}_{d}(N)$. The irreps
	indices of the matrix $R_{N}^{d}$ are ordered lexicographically and we set $	R_{N}^{N}=R_{N}$.
\end{definition}

The matrix $R_{N}^{d}$ has the following interesting properties:
	\begin{enumerate}[1)]
		\item The sum of $1^{\prime }s$ in a given row $\alpha $ is equal to the number of irreps $\psi^{\mu} \in \widehat{S}_{d}(N)$ included in the
		representation $\Ind_{S(N-1)}^{S(N)}(\varphi^{\alpha} )$.
		\item The sum of $1^{\prime }s$ in a given column $\mu $ is equal to the number of irreps $\varphi^{\alpha} \in \widehat{S}_{d}(N-1)$ included in the
		representation $\operatorname{Res}_{S(N-1)}^{S(N)}(\psi^{\mu})$.
		\item The number $1$ in the position $(\alpha ,\mu )$ in $R_{N}^{d}$ means that
		the projector $F_{\mu }(\alpha )$ is non-zero.
	\end{enumerate}

\begin{example}
	In this example we show the explicit form of matrix $R_N^d$ given in Definition~\ref{def_of_R} for $d=N=4$:
	\be
	R_{4}^{4}=\left( 
	\begin{array}{ccccc}
		1 & 1 & 0 & 0 & 0 \\ 
		0 & 1 & 1 & 1 & 0 \\ 
		0 & 0 & 0 & 1 & 1%
	\end{array}%
	\right) . 
	\ee
\end{example}

Matrices $R_{N}^{d}$ have the following property:
\begin{proposition}
	\label{Rprop6}
	For any $d\geq 2$ and $N\geq 2$ the matrix $R_{N}^{d}$ has maximal rank
	equal $|\widehat{S}_{d}(N-1)|$, so the rows of the matrix $R_{N}^{d}$ are
	linearly independent.
\end{proposition}

\begin{proof}
	Let consider a square submatrix of maximal dimension whose columns are
	indexed by irreps $\psi^{\mu} \in \widehat{S}_{d}(N)$
	\be
	\mu =\alpha+\Box, 
	\ee
	where the box is added to the first row of $\alpha$ which labels $\varphi^{\alpha} \in \widehat{S}_{d}(N-1)$
	, so Young diagrams $\mu$ are ordered similarly to $\alpha$. Then one can show that such a square matrix is
	upper triangular with $1^{\prime }s$ on the diagonal, therefore the
	corresponding minor of maximal dimension is non-zero.
\end{proof}

We now define two other matrices which are connected with $R_{N}^{d}$:
\begin{definition}
	\label{defGH}
	\be
	G_{N}^{d}\equiv (g_{\mu \nu }^{d}(N))=(R)_{N}^{d})^{T}R_{N}^{d}\in \mathbb{M}(\widehat{S}_{d}(N),\mathbb{Z}),
	\ee
	\be
	H_{N}^{d}\equiv (h_{\alpha \beta }^{d}(N))=R_{N}^{d}(R_{N}^{d})^{T}\in \mathbb{M}(\widehat{S}_{d}(N-1),\mathbb{Z}),
	\ee
	each of which is Gram matrix of the columns of the matrix $
	R_{N}^{d}$ and Gram matrix of the rows the matrix $R_{N}^{d}$ respectively. The matrix $G_{N}^{d}$ is indexed by Young diagrams $\mu$ such that $\psi^{\mu} \in \widehat{S}_{d}(N)$
	whereas the matrix $H_{N}^{d}$ is indexed by Young diagrams $\alpha$ such that $\varphi^{\alpha} \in \widehat{S}_{d}(N-1)
	$.
\end{definition}

From Proposition~\ref{Rprop6} it follows that the matrix $H_{N}^{d}$ is
	invertible with the following connection between the spectra of the matrices $G_{N}^{d}$ and $H_{N}^{d}$:

\begin{proposition}
	\label{Rprop9}
	All non-zero eigenvalues of the matrix $G_{N}^{d}$ are precisely the
	eigenvalues of the matrix $H_{N}^{d}$ and the corresponding eigenvectors are
	related by matrix $R_{N}^{d}$. In particular, matrices $G_{N}^{d}$
	and $H_{N}^{d}$ have the same spectral radius.
\end{proposition}
We now show that matrices $R_{N}^{d}$, $G_{N}^{d}$, and $
H_{N}^{d} $ are strictly connected with Teleportation Matrix $M_{F}^{d}(N)$ given in Definition~\ref{M_F}:

\begin{theorem}
	The following relation holds
	\be
	G_{N}^{d}=M_{F}^{d}(N), 
	\ee
	so the matrix $M_{F}^{d}(N)$ is in fact a Gram matrix.
\end{theorem}

\begin{proof}
	Let consider the matrix element of the matrix $G_{N}^{d}$ (we omitt here the
	index $N$)
	\be
	\label{rhs1}
	g_{\mu \nu }^{d}=\sum_{\alpha }r_{\mu \alpha }^{d}r_{\alpha \nu }^{d}.
	\ee
	If $\mu =\nu $, then the non-zero terms in the sum on RHS of~\eqref{rhs1} are those for $
	\alpha =\mu -\square $, so $h(\alpha )\leq h(\mu )$ and the summation of $%
	1^{\prime }s$ is over those $\alpha$ labelling $\varphi^{\alpha}\in \widehat{S}_{d}(N-1)$, from which
	one obtains $\mu$ by adding properly one box to and $\psi^{\mu} \in \widehat{S}_{d}(N)$.
	Therefore $g_{\mu \mu }^{d}=(M_{F}^{d})_{\mu \mu }$. 
	
	If $\mu \neq \nu$,
	then the non-zero terms in the sum on RHS of~\eqref{rhs1} are for such $\alpha$ labelling $\varphi^{\alpha}\in \widehat{S}_{d}(N-1)$, for which one obtains both $\mu ,\nu$ by adding one box to $\alpha $ and  $\psi^{\mu},\psi^{\nu}\in \widehat{S}_{d}(N)$.
	 There exists only one such Young diagram $\alpha$ and it means that the Young diagrams $\mu ,\nu $
	are such that one is obtained from another one by moving one box, which
	is a definition of the element $(M_{F}^{d})_{\mu \nu }$ in the matrix $
	M_{F}^{d}(N)$.
\end{proof}

\begin{corollary}
	For any $d\geq 2$ and $N\geq 2$ the matrix $M_{F}^{d}(N)$ is positive
	semi-definite.
\end{corollary}

Proving semi-definitenes of $M_{F}^{d}(N)$ becomes straightforward when we adopt the approach of this subsection.
To derive the remaining result we need the following simple observation:

\begin{remark}
	For any $d\geq 2$ and $N\geq 2$ the matrix $R_{N}^{d}$ is a principal
	submatrix of the full matrix $R_{N}$.
\end{remark}

As well as two technical lemmas below:

\begin{lemma}
	Fix two irreps $\varphi^{\alpha} \in \widehat{S}(N-1)$ and $\psi^{\mu}\in \widehat{S}(N)$.
	If a Young diagram $\alpha$ is such that $\alpha =\mu -\square$ i.e. 
	\be
	\alpha =(\alpha _{1},\ldots,\alpha _{i},\ldots,\alpha _{k})=(\mu _{1},\ldots,\mu
	_{i}-1,\ldots,\mu _{k}),
	\ee
	then $\gamma =(\mu _{1},\ldots,\mu _{i-1}-1,\mu _{i}-1,\ldots,\mu _{k})\vdash N-2$ is also a well defined Young diagram and it labels an irrep of $S(N-2)$.
\end{lemma}

\begin{lemma}
	Consider the matrix $R_{N}^{d}$ as a principal submatrix of the full matrix $%
	R_{N}$, then the row labelled by $\alpha :h(\alpha )<d$ of the
	submatrix $R_{N}^{d}$ includes all $1^{\prime }s$ from the row labelled by $\alpha $ in
	the matrix $R_{N}$. If the row labelled by $\alpha$ of the
	submatrix $R_{N}^{d}$ is such that $h(\alpha )=d$, then there is a single $1$, which is outside the submatrix $R_{N}^{d}$.
\end{lemma}

Using these statements one can prove the following important relation
between the matrices $M_{F}^{d}(N-1)$ and $H_{N}^{d}$

\begin{theorem}
	\label{thmR16}
	For any $d\geq 2$ and $N\geq 2$ we have 
	\be
	H_{N}^{d}=J_{p}+M_{F}^{d}(N-1),
	\ee
	where the matrix $J_{p}$ is of the form%
	\be
	J_{p}=\left( 
	\begin{array}{cc}
		\mathbf{1}_{p} & 0 \\ 
		0 & 0%
	\end{array}%
	\right) 
	\ee
	and $\mathbf{1}_{p}$ is the identity matrix of dimension $p$, which is the
	number of rows $\alpha$ for which $\varphi^{\alpha} \in \widehat{S}_{d}(N-1)$ of the submatrix $R_{N}^{d}$
	is such that $h(\alpha )=d$.
\end{theorem}

In particular we have 
\be
H_{N}^{2}=J_{1}+M_{F}^{2}(N-1),\qquad H_{N}^{N}=\mathbf{1}+M_{F}(N-1),
\ee
i.e. in the last case  $J_{p}$ is a identity matrix.

\begin{remark}
	The importance of Theorem~\ref{thmR16} follows from the fact that the matrices $%
	H_{N}^{d}$ and $G_{N}^{d}=M_{F}^{d}(N)$ have the same non-zero eigenvalues
	(see  Prop.~\ref{Rprop9} ), so the relation in the theorem yields a recursive formula between eigenvalues, matrices $M_{F}^{d}(N)$ and $M_{F}^{d}(N-1)$.
\end{remark}

The starting point of the recursive descent is the case $d=N$ which then gives a
following recursive relation for the maximal eigenvalues $\lambda _{\max
}(N)$ of matrices $M_{F}(N)$
\be
\lambda _{\max }(N)=1+\lambda _{\max }(N-1)\Rightarrow \lambda _{\max }(N)=N,
\ee
which coincides with the earlier result obtained using spectral decomposition of the matrix $M_{F}(N)$ but with significantly less effort.

\section{Optimisation over a resource state in the dPBT}\label{sec:sdp}
We now turn to the case when both the resource state $|\Psi\>$ and Alice's measurements $\{\Pi_a\}_{a=1}^N$ are optimised simultaneously. Since from~\cite{ishizaka_quantum_2009} we know that this problem can be cast in terms of SDP, we provide analytical solutions to both primal and dual SDPs obtaining optimal form of POVMs and the state $|\Psi\>$. By showing that the primal matches the dual, we obtain the optimal fidelity.  The optimal fidelity of the dPBT is directly expressed in terms of the Teleportation Matrix  $M_F$ given in Definition~\ref{M_F} or its principal matrices if the dimension $d$ is smaller than number of the ports $N$. More precisely, it is given by the square of a maximal eigenvalue divided by the square of the dimension of the teleported system.

Figure~\ref{F_plot} illustrates how optimal fidelity compares to previous results.

\begin{figure}[ht!]
	\centering
	\includegraphics[width=110mm]{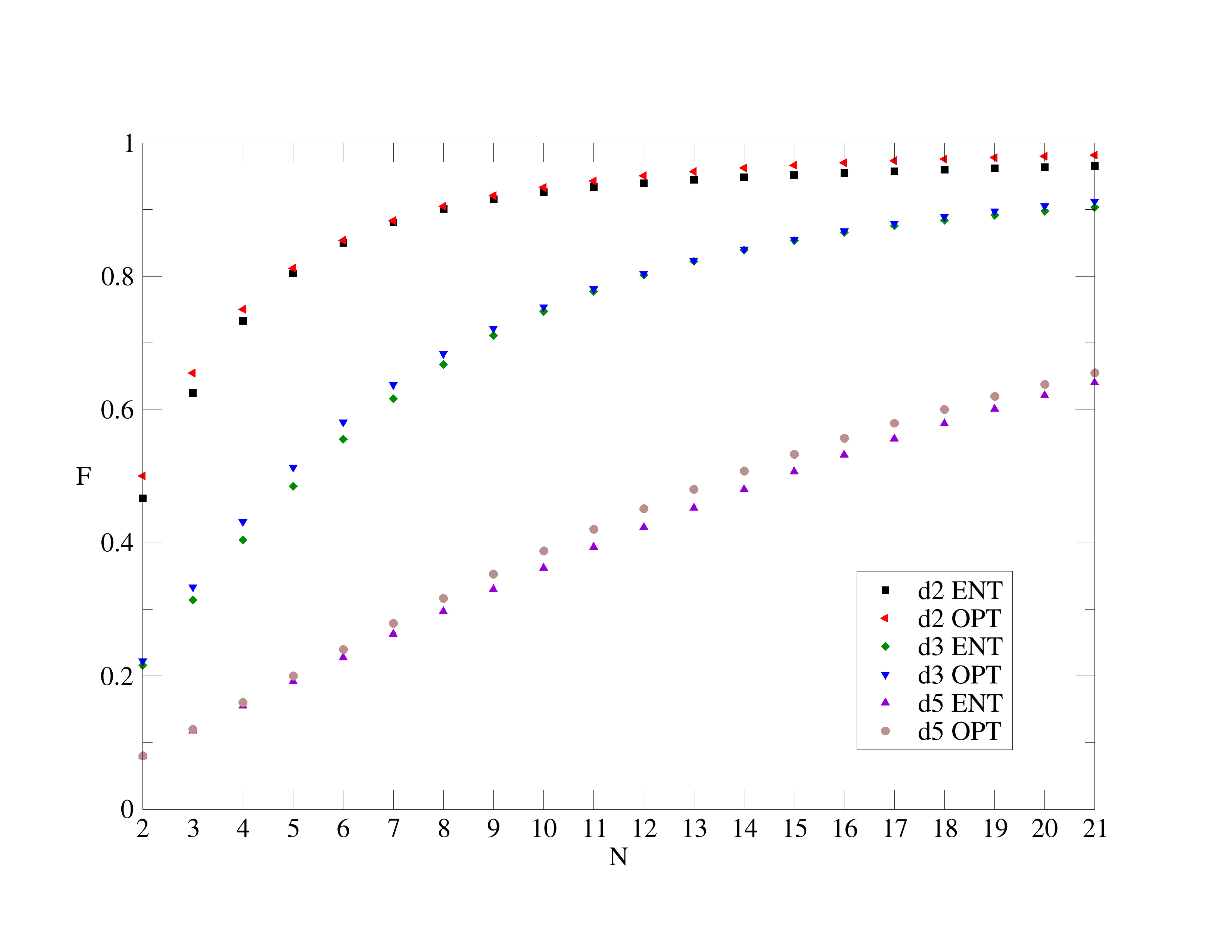}
	\caption{Best achievable fidelity of port-based teleportation when both the state and the measurement is optimized. dX ENT denotes the fidelity of the dPBT when the resource state consists of maximally entangled pairs and only measurement is optimized; X corresponds to the dimension of the teleported state. dX OPT denotes the best possible fidelity achieved by optimizing the resource state and measurement simultaneously}
	\label{F_plot}
\end{figure}

\subsection{The primal SDP problem}
The primal problem is to compute:
\be
\label{prim1}
F^*=\frac{1}{d^2}\max_{\{\Pi_i\}}\sum_{a=1}^N\tr\left[\Pi_a\sigma_a\right], 
\ee
with respect to constraints
\be
\label{bb1f}
(1)\quad \sum_{a=1}^N\Pi_a\leq X_A\ot \mathbf{1}_{\overline{B}},\quad (2)\quad \tr X_A=d^N.
\ee
In the above $\{\Pi_a\}_{a=1}^N$ is the set of POVMs used by Alice, and $X_A=O_A^{\dagger}O_A$, where $O_A$ is a global operation performed on Alices' half of the maximally entangled resource state. The solution of~\eqref{prim1} with the constraints~\eqref{bb1f} is given in the following
\begin{theorem}
	\label{main1}
	The quantity $F^*$ in the primal problem can be expressed as:
	\be
	F^*=\frac{1}{d^2}\left| \left| M_F\right| \right|_{\infty},
	\ee
	where $||M_F||_{\infty}$ denotes the infinity norm of the Teleportation Matrix $M_F$ is given in Definition~\ref{M_F}.
\end{theorem}
\begin{proof}
	Here we assume the most general form of the POVMs (indeed more general than in~\eqref{POVM1}); for $a=1,\ldots,N$ we take:
	\be
	\label{assum1}
	\Pi_a=\Pi \sigma_a \Pi,
	\ee
	with
	\be
	\label{assum2}
	\Pi=\sum_{\alpha}\sum_{\mu \in \alpha}p_{\mu}(\alpha)F_{\mu}(\alpha),  \quad  p_{\mu}(\alpha) \geq 0,
	\ee
	and 
	\be
	\label{Xa} \ 
	X_A=\sum_{\mu}c_{\mu}P_{\mu},\quad   c_{\mu}\geq 0.
	\ee
	 We rewrite expression~\eqref{prim1} using our assumption about the form of POVMs $\Pi_a$ for $a=1,\ldots,N$ given in~\eqref{assum1}:
	\be
	\begin{split}
		F^*&=\frac{1}{d^2}\max_{\{\Pi_a\}}\tr\left[\sum_{a=1}^N\Pi_a\sigma_i\right]=\frac{1}{d^2}\max_{\Pi}\sum_a\tr\left[\Pi \sigma_a \Pi \sigma_i \right]\\
		&=\frac{N}{d^2}\max_{\Pi}\tr\left[\Pi \sigma_N \Pi \sigma_N \right]=\frac{N}{d^{2N}}\max_{\Pi}\tr\left[\Pi (\mathbf{1}\ot P_+) \Pi ( \mathbf{1}\ot P_+)\right],
	\end{split}
	\ee
where we use the fact that $\tr\left[\Pi_a\sigma_a \right]$ does not depend on the index $a=1,\ldots,N$. This property allows us to compute the trace for fixed value $a=N$ and multiply it $N$ times. Here and further in this manuscript by  $P_+$ we denote projector onto the maximally entangled  state $|\Phi^+\>$ between $N-$th and $n-$th subsystem, and the identity operator $\mathbf{1}$ on $N-1$ first subsystems. Substituting decomposition of $\Pi$ given in~\eqref{assum2}, fact that $\mathbf{1}\ot P_+=\frac{1}{d}V^{t_n}(N,n)$, and decomposition~\eqref{op_F}  we write:	
\be
F^*=\frac{N}{d^{2N+2}}\max_{\{p_{\mu}(\alpha),p_{\mu'}(\alpha')\}}\sum_{\alpha,\alpha'}\sum_{\substack{\mu\in\alpha \\ \mu'\in \alpha'}}p_{\mu}(\alpha)p_{\mu'}(\alpha')\tr\left[M_{\alpha}P_{\mu}V^{t_n}(N,n)M_{\alpha'}P_{\mu'}V^{t_n}(N,n) \right]. 
\ee
Using that $V^{t_n}(N,n)M_{\alpha}=V^{t_n}(N,n)P_{\alpha}$ (see Fact 13 of~\cite{Stu2017}) we have
\be
F^*=\frac{N}{d^{2N+2}}\max_{\{p_{\mu}(\alpha),p_{\mu'}(\alpha')\}}\sum_{\alpha,\alpha'}\sum_{\substack{\mu\in\alpha \\ \mu'\in \alpha'}}p_{\mu}(\alpha)p_{\mu'}(\alpha')\tr\left[P_{\mu}V^{t_n}(N,n)P_{\alpha'}P_{\mu'}V^{t_n}(N,n)P_{\alpha} \right]. 
\ee	
Using properties $[P_{\alpha},V^{t_n}(N,n)]=0$, $[P_{\alpha},P_{\mu}]=0$, $P_{\alpha}P_{\alpha'}=\delta_{\alpha \alpha'}P_{\alpha}$, and again $V^{t_n}(N,n)M_{\alpha}=V^{t_n}(N,n)P_{\alpha}$ we reduce above expression to
\be
\begin{split}
F^*&=\frac{N}{d^{2N+2}}\max_{\{p_{\mu}(\alpha),p_{\mu'}(\alpha)\}}\sum_{\alpha}\sum_{\substack{\mu\in\alpha \\ \mu'\in \alpha}}p_{\mu}(\alpha)p_{\mu'}(\alpha)\tr\left[P_{\mu}V^{t_n}(N,n)P_{\alpha}P_{\mu'}V^{t_n}(N,n)P_{\alpha}\right]\\
&=\frac{N}{d^{2N}}\max_{\{p_{\mu}(\alpha),p_{\mu'}(\alpha)\}}\sum_{\alpha}\sum_{\substack{\mu\in\alpha \\ \mu'\in \alpha}}p_{\mu}(\alpha)p_{\mu'}(\alpha)\tr\left[F_{\mu}(\alpha)(P_{\alpha}\ot P_+)F_{\mu'}(\alpha)(P_{\alpha}\ot P_+) \right].
\end{split}
\ee 
In the next step we use of the identity operator in the form $\mathbf{1}=\sum_{\alpha}P_{\alpha}=\sum_{\alpha}\sum_{k=1}^{d_{\alpha}}\sum_{r=1}^{m_{\alpha}}|\varphi_{k,r}(\alpha)\>\<\varphi_{k,r}(\alpha)|$, where vectors $\{|\varphi_{k,r}(\alpha)\>\}_{k=1}^{d_{\alpha}}$ span $r$-th block of the irrep labelled by Young diagram $\alpha$:
	\be
	\begin{split}
		F^*&=\frac{N}{d^{2N}}\max_{\{p_{\mu}(\alpha), p_{\mu'}(\alpha)\}}\sum_{\alpha}\sum_{\substack{\mu \in \alpha \\ \mu'\in \alpha}}p_{\mu'}(\alpha)p_{\mu}(\alpha)\times \\ & \times \sum_{k,l=1}^{d_{\alpha}}\sum_{r,s=1}^{m_{\alpha}}\tr\left[F_{\mu}(\alpha) |\varphi_{k,r}(\alpha)\>\<\varphi_{k,r}(\alpha)\ot P_+|F_{\mu'}(\alpha) |\varphi_{l,s}(\alpha)\>\<\varphi_{l,s}(\alpha)|\ot P_+\right]\\
		&=\frac{N}{d^{2N}}\max_{\{p_{\mu}(\alpha), p_{\mu'}(\alpha)\}}\sum_{\alpha}\sum_{\substack{\mu \in \alpha \\ \mu'\in \alpha}}p_{\mu'}(\alpha)p_{\mu}(\alpha)\times \\ & \times
		\sum_{k,l=1}^{d_{\alpha}}\sum_{r,s=1}^{m_{\alpha}}\tr\left[ |\varphi_{l,s}(\alpha)\>\<\varphi_{k,r}(\alpha)|\ot P_+F_{\mu'}(\alpha) \right] \tr\left[|\varphi_{k,r}(\alpha)\>\<\varphi_{l,s}(\alpha)\ot P_+|F_{\mu}(\alpha) \right].
	\end{split}
	\ee
	Using Fact~\ref{f1} we can simplify above expression as
	\be
	\label{cos1}
	\begin{split}
		F^*&=\frac{N}{d^{2N+2}}\max_{\{p_{\mu}(\alpha), p_{\mu'}(\alpha)\}}\sum_{\alpha}\sum_{\substack{\mu \in \alpha \\ \mu'\in \alpha}}p_{\mu'}(\alpha)p_{\mu}(\alpha)\frac{m_{\mu'}m_{\mu}}{m_{\alpha}^2}\sum_{k,l=1}^{d_{\alpha}}\sum_{r,s=1}^{m_{\alpha}}\delta_{lk}^2\delta_{sr}^2\\
		&=\frac{N}{d^{2N+2}}\max_{\{p_{\mu}(\alpha), p_{\mu'}(\alpha)\}}\sum_{\alpha}\frac{d_{\alpha}}{m_{\alpha}}\sum_{\substack{\mu \in \alpha \\ \mu'\in \alpha}}p_{\mu'}(\alpha)p_{\mu}(\alpha)m_{\mu'}m_{\mu}\\
		&=\frac{N}{d^{2N+2}}\max_{\{p_{\mu}(\alpha)\}}\sum_{\alpha}\frac{d_{\alpha}}{m_{\alpha}}\left(\sum_{\mu \in \alpha}p_{\mu}(\alpha)m_{\mu} \right)^2.
	\end{split}
	\ee
	Form the definition of $\Pi$ we see that $\forall \pi\in S(N) \ [\Pi,V(\pi)]=0$. Together with~\eqref{rodec} we write
	\be
	\sum_{a=1}^N\Pi_a=\Pi\sum_{a=1}^N\sigma_a\Pi=\Pi \rho \Pi=\Pi^2\rho=\sum_{\alpha}\sum_{\mu \in \alpha}p_{\mu}^2(\alpha)\lambda_{\mu}(\alpha)F_{\mu}(\alpha).
	\ee
	Similarly to Eqn.(37) in~\cite{ishizaka_quantum_2009} we get
	\be
	\sum_{a=1}^N\Pi_a=\sum_{\alpha}\sum_{\mu \in \alpha}p_{\mu}^2(\alpha)\lambda_{\mu}(\alpha)F_{\mu}(\alpha)=\sum_{\mu}\sum_{\alpha \in \mu}p_{\mu}^2(\alpha)\lambda_{\mu}(\alpha)F_{\mu}(\alpha)\leq \sum_{\mu}c_{\mu}P_{\mu} \ot \mathbf{1}_{n}.
	\ee
	Note that $F_{\mu}(\alpha) \subset P_{\mu}$, so we have $p_{\mu}^2(\alpha)\lambda_{\mu}(\alpha) \leq c_{\mu}$. Now we see that
	the fidelity $F^*$ given by expression~\eqref{cos1}  can only increase, when we increase coefficients $p_{\mu}(\alpha)$.
	Thus for any fixed  $c_\mu$ it is optimal to choose $p_{\mu}(\alpha)$ satisfying 
	\be
	\label{imp1}
	\forall \alpha \quad p_{\mu}^2(\alpha)\lambda_{\mu}(\alpha)=c_{\mu}.
	\ee
	Finally from the  normalisation condition (expression (2) of~\eqref{bb1f}) and by substitution of~\eqref{Xa} we get constraint on coefficients $c_{\mu}$
	\be
	\label{www}
	\tr X_A=\sum_{\mu}c_{\mu}\tr P_{\mu}=\sum_{\mu}c_{\mu}d_{\mu}m_{\mu}=d^N.
	\ee
	Taking $v_{\mu}^2=\frac{1}{d^N}c_{\mu}d_{\mu}m_{\mu}$ together with the equation ensuring maximal possible value of the quantity $F^*$  given in ~\eqref{imp1} we write
	\be
	p_{\mu}^2(\alpha)\lambda_{\mu}(\alpha)d_{\mu}m_{\mu}=\left(\frac{1}{d^N}c_{\mu}d_{\mu}m_{\mu} \right)d^N=d^Nv_{\mu}^2. 
	\ee
	Using the explicit formula for $\lambda_{\mu}(\alpha)$ we can compute $p_{\mu}(\alpha)$ in terms of new coefficients $v_{\mu}$ as
	\be
	\label{forp}
	p_{\mu}(\alpha)=\frac{d^{N}}{\sqrt{N}}\sqrt{\frac{m_{\alpha}}{d_{\alpha}}}\frac{v_{\mu}}{m_{\mu}}.
	\ee
	Now inserting above formula into~\eqref{cos1} we have
	\be
	\label{FF2}
	F^*=\frac{N}{d^{2N+2}}\max_{\{v_{\mu}\}}\sum_{\alpha}\frac{d_{\alpha}}{m_{\alpha}}\left(\sum_{\mu \in \alpha}\frac{d^N}{\sqrt{N}} \sqrt{\frac{m_{\alpha}}{d_{\alpha}}}\frac{v_{\mu}}{m_{\mu}}m_{\mu}\right)^2=\frac{1}{d^2} \max_{\{v_{\mu}\}}\sum_{\alpha}\left(\sum_{\mu \in \alpha}v_{\mu}\right)^2. 
	\ee
	Using equation~\eqref{www} we get
	\be
	\label{pcon1bb}
	d^N\left(\sum_{\mu}\frac{1}{d^N}c_{\mu}d_{\mu}m_{\mu} \right)=d^N\sum_{\mu}v_{\mu}^2=d^N \Rightarrow \sum_{\mu}v_{\mu}^2=1.
	\ee
	The above condition is just a normalisation condition for some vector $v$, i.e. $||v||^2=\sum_{\mu}v_{\mu}^2=1$.
	Finally writing  more explicitly the double sum in~\eqref{FF2} we see the following
	\be
	\label{rew12}
	\sum_{\alpha}\left(\sum_{\mu \in \alpha}v_{\mu} \right)^2=\sum_{\alpha}\left(\sum_{\mu \in \alpha}v_{\mu}^2+\sum_{\substack{\mu \neq \nu \\ \mu,\nu\in \alpha}}v_{\mu}v_{\nu} \right)=\sum_{\mu}n_{\mu}v_{\mu}^2+ \sum_{\substack{\mu \neq \nu \\ \mu/\nu =\Box}}v_{\mu}v_{\nu},
	\ee
	where $n_{\mu}$ is number of $\alpha \vdash N-1$ for which $\mu \in \alpha$.  Having expression~\eqref{rew12} together with~\eqref{pcon1bb} we rewrite the equation~\eqref{FF2} as
	\be
	F^*=\frac{1}{d^2}\max_{v : ||v||=1}\<v|M_F|v\>\equiv \frac{1}{d^2}||M_F||_{\infty},
	\ee
\end{proof}

\subsection{The dual SDP problem}
\label{dualSDP}
The dual problem is to compute:
\be
\label{d1}
F_*=d^{N-2}\min_{\Omega}\left| \left|\tr_B \Omega\right| \right|_{\infty} ,
\ee
with respect to constraints
\be
\label{d2}
\Omega -\sigma_a \geq 0, \quad a=1,\ldots,N.
\ee
In the above $\Omega$ is an arbitrary operator acting on $N$ subsystems. The solution of~\eqref{d1} with the constraints defined in~\eqref{d2} is given in the following
\begin{theorem}
	\label{main2}
	The quantity $F_*$ in the dual problem can be expressed as:
	\be
	F_*=\frac{1}{d^2}\left| \left| M_F\right| \right|_{\infty},
	\ee
	where $||M_F||_{\infty}$ denotes the infinity norm of the Teleportation Matrix $M_F$ is given in Definition~\ref{M_F}.
\end{theorem}
\begin{proof}
	Assume the general form of the operator which gives contribution to $F_*$ as
	\be
	\label{omtilde}
	\widetilde{\Omega}=\sum_{\alpha \vdash N-1}\widetilde{\Omega}(\alpha)=\sum_{\alpha \vdash N-1}\sum_{\mu \in \alpha} \omega_{\mu}(\alpha)F_{\mu}(\alpha),\quad \omega_{\mu}(\alpha)\geq 0.
	\ee
	By choosing coefficients $\omega_{\mu}(\alpha)$ we ensure that $\widetilde{\Omega} - \sigma_a \geq 0$ for $a=1,\ldots,N$, where $\sigma_a=\frac{1}{d^{N-1}}\mathbf{1}_{\overline{an}}\ot P^+_{a,n}$ (see condition~\eqref{d2}), and $P_{a,n}^+$ is projector onto the maximally entangled state $|\Phi^+\>_{a,n}$ between $a-$th and $n-$th subsystem. Due to symmetry it is enough to check it only for $a=N$, and on all irreps $\alpha$. 
	\be
	\widetilde{\Omega} \geq \sigma_N \iff \forall \alpha  \quad \widetilde{\Omega}(\alpha) \geq \sigma_N P_{\alpha},
	\ee
	where $P_{\alpha}$ denotes a Young projector onto irrep labelled by the Young diagram $\alpha \vdash N-1$.
	More explicitly using form of the operator $\widetilde{\Omega}(\alpha)$ from~\eqref{omtilde} and resolution of the identity in terms of Young projectors $P_{\alpha}$ we have
	\be
	\forall \alpha \vdash N-1 \quad d^{N-1}\sum_{\mu \in \alpha}\omega_{\mu}(\alpha)F_{\mu}(\alpha)\geq P_{\alpha}\ot P_+.
	\ee
	We now ask when above condition is fulfilled. Form~\cite{Lewy1} we know, that
	\be
	A(\alpha)-\frac{1}{c(\alpha)}R(\alpha)\geq 0 \quad \text{if} \quad  c(\alpha)=\sum_{k=1}^{d_{\alpha}}\sum_{l=1}^{m_{\alpha}}\<\Phi_+|\<\varphi_{k,l}(\alpha)|A^{-1}(\alpha)|\varphi_{k,l}(\alpha)\>|\Phi_+\>,
	\ee
	where for fixed $l=1,\ldots,m_{\alpha}$ vectors $|\varphi_{k,l}(\alpha)\>$ span one irrep of $S(N-1)$ labelled by Young diagram $\alpha$ and 
	\be
	A(\alpha)=d^{N-1}\sum_{\mu \in \alpha}\omega_{\mu}(\alpha)F_{\mu}(\alpha), \quad R(\alpha)=P_{\alpha}\ot P_+.
	\ee
	Having above we are in the position to compute the constant $c(\alpha)$ for all irreps $\alpha$
	\be
	\begin{split}
		c(\alpha)&=\frac{1}{d^{N-1}}\sum_{k=1}^{d_{\alpha}}\sum_{l=1}^{m_{\alpha}}\<\Phi_+|\<\varphi_{k,l}(\alpha)|\sum_{\mu \in \alpha}\omega_{\mu}^{-1}(\alpha)F_{\mu}(\alpha)|\varphi_{k,l}(\alpha)\>|\Phi_+\>\\
		&=\frac{1}{d^{N-1}}\sum_{\mu\in\alpha}\omega_{\mu}^{-1}(\alpha)\sum_{k=1}^{d_{\alpha}}\sum_{l=1}^{m_{\alpha}}\tr\left[ |\varphi_{k,l}(\alpha)\>\<\varphi_{k,l}(\alpha)|\ot P_+F_{\mu}(\alpha)\right]\\
		&=\frac{1}{d^N} \sum_{\mu\in\alpha}\omega_{\mu}^{-1}(\alpha)\frac{m_{\mu}}{m_{\alpha}},
	\end{split}
	\ee
	since we used Fact~\ref{f1} from Appendix~\ref{App1}. Now, redefining the operator $\widetilde{\Omega}(\alpha)$ as
	\be
	\label{ch0}
	\begin{split}
		\Omega(\alpha)\equiv c(\alpha)\widetilde{\Omega}(\alpha)&=\frac{1}{d^N} \sum_{\nu\in\alpha}\omega_{\nu}^{-1}(\alpha)\frac{m_{\nu}}{m_{\alpha}}\sum_{\mu\in\alpha}\omega_{\mu}(\alpha)F_{\mu}(\alpha)\\
		&=\frac{1}{d^N} \sum_{\nu\in\alpha}\sum_{\mu\in\alpha} \frac{m_{\nu}\omega_{\mu}(\alpha)}{m_{\alpha}\omega_{\nu}(\alpha)}F_{\mu}(\alpha)
	\end{split}
	\ee
	we satisfy the constraint $\Omega - \sigma_N \geq 0$, since $\Omega=\sum_{\alpha}\Omega(\alpha)$. 
	In the next step we compute the  quantity $d^{N-2}\tr_n \Omega$ form~\eqref{d1}
	\be
	\begin{split}
		d^{N-2}\tr_n \Omega&=\frac{1}{d^2}\sum_{\alpha}\sum_{\nu\in\alpha}\sum_{\mu\in\alpha} \frac{m_{\nu}\omega_{\mu}(\alpha)}{m_{\alpha}\omega_{\nu}(\alpha)}\tr_n F_{\mu}(\alpha)=\frac{1}{d^2}\sum_{\alpha}\sum_{\nu\in\alpha}\sum_{\mu\in\alpha} \frac{m_{\nu}\omega_{\mu}(\alpha)}{m_{\mu}\omega_{\nu}(\alpha)}P_{\mu}\\
		&=\frac{1}{d^2}\sum_{\alpha}\sum_{\mu\in\alpha}\frac{\sum_{\nu\in \alpha}t_{\nu}(\alpha)}{t_{\mu}(\alpha)}P_{\mu}=\frac{1}{d^2}\sum_{\mu}\sum_{\alpha \in \mu}\frac{\sum_{\nu\in \alpha}t_{\nu}(\alpha)}{t_{\mu}(\alpha)}P_{\mu},
	\end{split}
	\ee
	where
	\be
	\label{ch1}
	t_{\mu}(\alpha)\equiv \frac{m_{\mu}}{\omega_{\mu}(\alpha)}.
	\ee
	From definition of $t_{\mu}(\alpha)$ we have to exclude all coefficients $\omega_{\mu}(\alpha)$ which are equal to zero from the decomposition~\eqref{omtilde} .
	Finally, the quantity $F_*$ in the dual problem given in~\eqref{d1} is given as
	\be
	\label{fff}
	\begin{split}
		F_*=d^{N-2}\min_{\Omega}||\tr_n\Omega||_{\infty}=\frac{1}{d^2}\min_{\{t_{\mu}(\alpha)\}}\max_{\mu}\sum_{\alpha \in \mu}\frac{\sum_{\nu\in \alpha}t_{\nu}(\alpha)}{t_{\mu}(\alpha)}.
	\end{split}
	\ee
	Since we are looking for the feasible solution we assume that $\forall \alpha \  \forall \mu\in\alpha \ t_{\mu}(\alpha)=t_{\mu}$:
	\be
	\label{ff}
	\forall \mu \vdash N \quad \sum_{\alpha\in\mu}\frac{\sum_{\nu\in\alpha}t_{\nu}}{t_{\mu}}=\frac{\sum_{\nu}\left(M_F \right)_{\mu \nu}t_{\nu}}{t_{\mu}},
	\ee
	where matrix $M_F$ is given in Definition~\ref{M_F}.
	Substituting~\eqref{ff} into~\eqref{fff} we reduce $min-max$ problem to
	\be
	F_*=\frac{1}{d^2}\min_{\{t_{\mu}\}}\max_{\mu}\frac{\sum_{\nu}\left(M_F \right)_{\mu \nu}t_{\nu}}{t_{\mu}}.
	\ee
	Consider the eigenproblem for the matrix $M_Ft=\lambda  t$, where $ t=(t_{\mu})$, and $\lambda \geq 0$, since $M_F$ is positive semi-definite. Writing eigenproblem for $M_F$ in the coordinates we have
	\be
	\forall \mu \vdash N \quad \sum_{\nu}\left(M_F \right)_{\mu \nu}t_{\nu}=\lambda t_{\mu} \  \Rightarrow \ \lambda=\frac{\sum_{\nu}\left(M_F \right)_{\mu \nu}t_{\nu}}{t_{\mu}}.
	\ee 
	Taking minimization over all vectors $t$ and maximal possible value over all allowed Young diagram $\mu$ we get definition of  the maximal eigenvalue of the matrix $M_F$:
	\be
	F_*=\frac{1}{d^2}\min_{\{t_{\mu}\}}\max_{\mu}\frac{\sum_{\nu}\left(M_F \right)_{\mu \nu}t_{\nu}}{t_{\mu}}=\frac{1}{d^2}||M_F||_{\infty}.
	\ee
\end{proof}
From Theorem~\ref{main1} and Theorem~\ref{main2} we get:
\begin{proposition}\label{opt_state}
	\begin{itemize}
		\item From equality $F^*=F_*$ we find that 
		\be
		F_{opt}=\frac{1}{d^2}||M_F||_{\infty}
		\ee
		is an optimal value of the fidelity in the case of the dPBT, where $M_F$ is Teleportation Matrix diven in Definition~\ref{M_F}.
		\item The optimal POVMs $\Pi_i=\Pi \sigma_i \Pi$ for $i=1,\ldots,N$ where $\Pi$ are given as:
		\be
		\label{opt_povm}
		\Pi=\frac{d^{N}}{\sqrt{N}}\sum_{\alpha}\sum_{\mu\in\alpha}\sqrt{\frac{m_{\alpha}}{d_{\alpha}}}\frac{v_{\mu}}{m_{\mu}}F_{\mu}(\alpha),
		\ee
		where the $\sigma_i$ is from~\eqref{sigma0}. The coefficients $v_{\mu}$ are the components of the eigenvector $v$ corresponding to the maximal eigenvalue of the Teleportation Matrix $M_F$ when $d\geq N$ or respective principal submatrix of $M_F$ otherwise.
		\item The optimal resource state $|\Psi\>$:
		\be
		\label{opt_State}
		|\Psi\>=\left(O_A\otimes \mathbf{1}_B \right)|\psi^+\>_{A_1B_1}\ot |\psi^+\>_{A_2B_2}\ot \cdots \ot |\psi^+\>_{A_NB_N},
		\ee
		where
		\be
		\label{pot_O}
		O_A=\sqrt{d^N}\sum_{\mu}\frac{v_{\mu}}{\sqrt{d_{\mu}m_{\mu}}}P_{\mu}.
		\ee
		In the above $P_{\mu}$ denotes Young projector onto irrep labelled by the Young diagram $\mu \vdash N$.
	\end{itemize}
\end{proposition}
\begin{proof}
Taking~\eqref{assum2} together with~\eqref{forp} we obtain desired form of operator $\Pi$. To obtain expression~\eqref{pot_O} we use~\eqref{Xa} with the condition $X_A=O_A^{\dagger}O_A$.
\end{proof}
In the regime $d\leq N$ from Proposition~\ref{c8} of Section~\ref{central} we can give a simple formula for optimal fidelity $F_{opt}$ in the dPBT: 
\be
\label{Ndd}
F_{opt}=\frac{N}{d^2},
\ee
since in this particular case $||M_F||_{\infty}=N$. We can run the same analysis for the eigenvector $v=(v_{\mu})$: when $d\geq N$ we know its analytical form as long as we assume that the respective characters of the irreps of $S(N)$ are given. In this case such vector is given as a column of the reduced character matrix $T=(\chi_{\mu}(C))$ introduced in Definition~\ref{reducedT} of Section~\ref{central}. We can construct it explicitly due to item $4$ in Corollary~\ref{c8}.
When $d<N$ we do not have analytical expressions (except for the qubit case discussed below) for the infinity norm of the principal submatrices of $M_F$ or eigenvector $v$.
In this case we use the algorithm presented in the Section~\ref{algo}.

The method of construction of the explicit matrix representation of the optimal POVMs and the state in the computational basis is described in detail in Appendix~\ref{natural}.

At the end of this section we discuss the asymptotic behaviour of the optimal fidelity $F_{opt}=F_{opt}(N,d)$ when number of ports $N$ tends to infinity with fixed local dimension of the Hilbert space $d$. From Corollary~\ref{cor8} and from well known relation $r(A)\leq ||A||$,  where $r(A)\equiv ||A||_{\infty }$ is the spectral radius of $0\leq A=(a_{ij})\in \M(n,\mathbb{C})$, and $||\cdot||$ is any matrix norm we get that fidelity $F_{opt}(N,d)$ is bounded in the following way
\be
\label{F1ee}
\forall N,d \quad F_{opt}(N,d)\leq 1,
\ee
which certifies our calculations. Denote by $\widetilde{F}_{ent}=\widetilde{F}_{ent}(N,d)$ the lower bound for the fidelity in the non-optimised case, when the resource state is a tensor product of $N$ $d-$dimensional singlets (see~\cite{beigi_konig})
\be
\widetilde{F}_{ent}=\frac{N}{d^2+N-1}.
\ee
We thus have $\widetilde{F}_{ent}(N,d)\leq F_{opt}(N,d)$. Moreover, for a fixed dimension $d$ we have $\lim_{N\rightarrow \infty}\widetilde{F}_{ent}(N,d)=1$, so together with expression~\eqref{F1ee} we see that  $\lim_{N\rightarrow \infty}F_{opt}(N,d)=1$.

\subsection{Comparison with known results}
In this section we compare our results to the only previously investigated case of $d=2$ from~\cite{ishizaka_remarks_2015,ishizaka_asymptotic_2008,ishizaka_quantum_2009}. We show how our approach relates to the latter when it comes to determining optimal fidelity and optimal POVMs with known representation of the dPBT. Moreover, we show how extending to higher dimensions of the underlying local Hilbert space reproduces the expression for the fidelity of the teleported state in the case of the maximally entangled resource state presented in~\cite{Stu2017}. The proof presented here, remarkably, does not require notion of partially reduced irreps which was indispensable in the previous approach of~\cite{Stu2017}.

We start from showing how the optimal fidelity $F_{opt}$ given in Proposition~\ref{opt_state} from Section~\ref{dualSDP} reduces to the results presented in earlier works. Whenever $N>2, d=2$ Proposition~\ref{c8} from Section~\ref{central} is not applicable since not all irreps of $S(N)$ appear. We thus cannot use the analytical formula for the optimal fidelity given by~\ref{Ndd}, and instead have to carry out the analysis of the infinity norm of principal submatrices of $M_F$. Fortunately, for this case principal submatrices of  $M_F$ (we absorb coefficient $1/4$ into definition of $M_F$) reduce to so-called tridiagonal matrix of the form
\be
M_F=\frac{1}{4}\begin{pmatrix}
	-x_1+b & c & 0 & 0 & \cdots & 0 & 0\\\
	a & b & c & 0 & \cdots & 0 & 0\\\
	0 & a & b & c & \cdots & 0 & 0\\
	\vdots & \vdots & \vdots & \vdots & \ddots & \vdots & \vdots\\
	0 & 0 & 0 & 0 & \cdots & a & -x_2+b
\end{pmatrix} \in \M(t,\mathbb{R}),
\ee
for which analytical expressions for eigenvalues are known;  $t$ is the number of allowed Young diagrams of $N$ for $d=2$, $a=c=1$, and $b=2$. The coefficients $x_1,x_2$ depend on the parity of $N$. Let us consider them separately.
\begin{enumerate}[a)]
	\item $x_1=1$ and $x_2=0$ when $N$ is odd.
\end{enumerate}
In this case from~\cite{Yueh} (Theorem 1, page 72) we know that all eigenvalues of $M_F$ for $k=1,\ldots,t$ are of the form:
\be
\begin{split}
	\lambda_k=\frac{1}{4}\left[ b+2\sqrt{ac}\cos\left(\frac{2k\pi}{2t+1} \right)\right] =\frac{1}{2}\left[1+\cos\left(\frac{2k\pi}{2t+1}\right) \right]=\cos^2\left(\frac{k\pi}{2t+1} \right),
\end{split}
\ee
since $\cos(2y)=2\cos^2y-1$. When $N$ is odd matrix $M_F$ is $(N+1)/2-$dimensional, so
\be
\label{form1}
\lambda_k=\cos^2\left(\frac{k\pi}{2\left(\frac{N+1}{2} \right)+1} \right) =\cos^2\left(\frac{k\pi}{N+2} \right),\qquad k=1,\ldots,(N+1)/2. 
\ee
\begin{enumerate}[b)]
	\item $x_1=x_2=1$ when $N$ is even.  
\end{enumerate}
In this case from~\cite{Yueh} (Theorem 4, page 73) we know that all eigenvalues of $M_F$ for $k=1,\ldots,t$ are of the form
\be
\lambda_k=\frac{1}{4}\left[ b+2\sqrt{ac}\cos\left(\frac{k\pi}{t} \right)\right] =\frac{1}{2}\left[1+\cos\left(\frac{k\pi}{t} \right)  \right]=\cos^2\left(\frac{k\pi}{2t} \right).
\ee
When $N$ is even matrix $M_F$ is $N/2+1-$dimensional, so
\be
\label{form2}
\lambda_k=\cos^2\left(\frac{k\pi}{2\left(\frac{N}{2}+1 \right) } \right)=\cos^2\left(\frac{k\pi}{N+2} \right), \qquad k=1,\ldots,N/2+1. 
\ee
In both cases, i.e. when $N$ is odd or even the maximal eigenvalue is obtained for $k=1$, and then optimal fidelity $F_{opt}$  is equal to:
\be
F_{opt}=||M_F||_{\infty}=\cos^2\left(\frac{\pi}{N+2} \right).
\ee
We see that the above expression reproduces optimal fidelity in Eqn. (41) from~\cite{ishizaka_quantum_2009}.

We now turn to the connection between our optimal POVMs and those derived in~\cite{ishizaka_quantum_2009} where authors propose the following optimal POVMs
\be
\label{POVM1}
\widetilde{\Pi}_a=\sum_{s=s_{\min}}^{(N-1)/2}z(s)\rho(s)^{-1/y(s)}\sigma_a(s)\rho(s)^{-1/y(s)},\quad a=1,\ldots,N,
\ee
where $s$ is the total spin number, and $z(s),y(s)$ some constant numbers for fixed $s$. This expression is valid only for the qubit case, but it can be easily translate into language of the irreps of $S(N)$ and all $d\geq 2$. Assume the general form of the optimal POVM to be
\be
\label{POVM2}
\widetilde{\Pi}_a=\sum_{\alpha \vdash N-1}z(\alpha)\rho(\alpha)^{-1/y(\alpha)}\sigma_a(\alpha)\rho(\alpha)^{-1/y(\alpha)},\quad a=1,\ldots,N,
\ee
where sum runs over all irreps labelled by Young diagrams of $N$ whose height is not greater than dimension $d$ of the underlying local Hilbert space. 
Now we are in the position to present direct connection between the most general decomposition of POVMs presented in \eqref{assum1},\eqref{assum2} and the form given in \eqref{POVM2}.
\begin{corollary}
	\label{t1}
	Having decompositions of POVMs defined in~\eqref{assum1},\eqref{assum2}, and~\eqref{POVM2} by comparison we can write the following equality between coefficients $p_{\mu}(\alpha)$ and $z(\alpha)$:
	\be
	\label{comp1}
	\begin{split}
		p_{\mu}(\alpha)&=\sqrt{z(\alpha)}\lambda_{\mu}(\alpha)^{-1/y(\alpha)}.
	\end{split}
	\ee
	In particular for $d=2$ we have a direct translation between optimal POVMs presented in~\cite{ishizaka_asymptotic_2008,ishizaka_quantum_2009,ishizaka_remarks_2015} (or see~\eqref{POVM1}) and the decomposition presented in this manuscript.
\end{corollary}
The equation~\eqref{comp1} can be obtained by direct comparison of~\eqref{assum1},~\eqref{assum2} with the expression~\eqref{POVM2} and fact that $\rho=\sum_{\alpha}\sum_{\mu\in\alpha}\lambda_{\mu}(\alpha)F_{\mu}(\alpha)$.

Before we go further and prove that the choice of the POVMs given in~\eqref{POVM2} reproduces correct expression for the fidelity in the dPBT in the case of the maximally entangled resource state we need the following auxiliary lemma 
\begin{lemma}
	\label{l1}
	The fidelity of the teleported state with the POVMs given from~\eqref{POVM2} is given by
	\be
	\label{l1a}
	F=\frac{1}{d^{N+1}}\sum_{\alpha \vdash N-1}z(\alpha)c(\alpha,y(\alpha))\tr\left[\rho(\alpha)^{1-1/y(\alpha)} \right],
	\ee 
	where 
	\be
	c(\alpha,y(\alpha))=\frac{1}{d}\sum_{\mu \in \alpha}\lambda_{\mu}(\alpha)^{-1/y(\alpha)}\frac{m_{\mu}}{m_{\alpha}}.
	\ee
\end{lemma}

\begin{proof}
	From~\cite{ishizaka_quantum_2009} we know that fidelity $F$ in the deterministic version of the protocol is given by
	\be
	\label{defF}
	F=\frac{1}{d^2}\tr\left[\sum_{a=1}^N\Pi_a \sigma_a\right],
	\ee
	where $\Pi_a$ are the POVMs given in~\eqref{POVM2}. Using explicit form of POVMs we get:
	\be
	\begin{split}
		F&=\frac{1}{d^2}\sum_{a=1}^{N}\tr\left[\sum_{\alpha}z(\alpha)\rho(\alpha)^{-1/y(\alpha)}\sigma_a(\alpha)\rho(\alpha)^{-1/y(\alpha)}\sigma_a(\alpha) \right]\\
		&=\frac{N}{d^{2N}}\sum_{\alpha}z(\alpha)\tr\left[\rho(\alpha)^{-1/y(\alpha)} P_{\alpha}\ot P_+\rho(\alpha)^{-1/y(\alpha)}P_{\alpha}\ot P_+ \right].
	\end{split}
	\ee
	We used the fact that due to symmetry the trace in~\eqref{defF}  does not depend on the index $i$ and that $\sigma_N(\alpha)=P_{\alpha}\ot P_+$. Using the decomposition of the Young projector $P_{\alpha}=\sum_{k=1}^{d_{\alpha}}\sum_{r=1}^{m_{\alpha}}|\varphi_{k,r}(\alpha)\>\<\varphi_{k,r}(\alpha)|$ we have
	\be
	\label{F2}
	\begin{split}
		F&=\frac{N}{d^{2N}}\sum_{\alpha}z(\alpha)\sum_{k,l=1}^{d_{\alpha}}\sum_{r,p=1}^{m_{\alpha}}\tr\left[ |\varphi_{k,r}(\alpha)\>\<\varphi_{k,r}(\alpha)|\ot P_+\rho(\alpha)^{-1/y(\alpha)} |\varphi_{l,p}(\alpha)\>\<\varphi_{l,p}(\alpha)|\ot P_+\rho(\alpha)^{-1/y(\alpha)} \right]\\
		&=\frac{N}{d^{2N}}\sum_{\alpha}z(\alpha)\sum_{k,l=1}^{d_{\alpha}}\sum_{r,p=1}^{m_{\alpha}}\<\Phi_+|\<\varphi_{k,r}(\alpha)|\rho^{-1/y(\alpha)}|\Phi_+\>|\varphi_{l,p}(\alpha)\>\tr\left[ |\varphi_{k,r}(\alpha)\>\<\varphi_{l,p}(\alpha)|\ot P_+\rho^{-1/y(\alpha)}\right].
	\end{split}
	\ee
	Using Remark~\ref{def_c} and with some simplification we get
	\be
	\label{qwe}
	\begin{split}
		F&=\frac{N}{d^{2N}}\sum_{\alpha}z(\alpha)c(\alpha,y(\alpha))\sum_{k=1}^{d_{\alpha}}\sum_{r=1}^{m_{\alpha}}\tr\left[ |\varphi_{k,r}(\alpha)\>\<\varphi_{k,r}(\alpha)|\ot P_+\rho(\alpha)^{-1/y(\alpha)} \right] \\
		&=\frac{N}{d^{N+1}}\sum_{\alpha}z(\alpha)c(\alpha,y(\alpha))\tr\left[ \frac{P_{\alpha}}{d^{N-1}}\ot P_+\rho(\alpha)^{-1/y(\alpha)} \right]\\
		&=\frac{N}{d^{N+1}}\sum_{\alpha}z(\alpha)c(\alpha,y(\alpha))\tr\left[ \sigma_{N}(\alpha)\rho(\alpha)^{-1/y(\alpha)} \right]\\
		&=\frac{1}{d^{N+1}}\sum_{\alpha}z(\alpha)c(\alpha,y(\alpha))\tr\left[\sum_{a=1}^N  \sigma_{a}(\alpha)\rho(\alpha)^{-1/y(\alpha)} \right]\\
		&=\frac{1}{d^{N+1}}\sum_{\alpha}z(\alpha)c(\alpha,y(\alpha))\tr\left[\rho(\alpha)^{1-1/y(\alpha)} \right]. 
	\end{split}
	\ee
	In~\eqref{qwe} we used the fact that $\rho(\alpha)=\sum_{a=1}^N  \sigma_{a}(\alpha)=\sum_{a=1}^N P_{\alpha}/d^{N-1}\ot P_{a,n}^+$, where $P_{a,n}^+$ is the projector on the maximally entangled state $|\Phi^+\>_{a,n}$ between $a-$th and $n-th$ system.
\end{proof}
We do not claim yet that POVMs given by~\eqref{POVM2} are indeed the optimal ones for any $d\geq 2$. We only derived the formula for the fidelity of the teleported state for this particular choice of measurements. Now using above lemma we can show that 

\begin{lemma}
	\label{CF_max}
	Substituting in expression~\eqref{l1a} of Lemma~\ref{l1} and eq.~\eqref{POVM1} $\forall \alpha \ y(\alpha)=2$ and $z(\alpha)=1$ we reproduce POVMs (square root measurement) and fidelity in the dPBT in the case of the maximally entangled state as a resource state.
\end{lemma}

\begin{proof}
	Inserting $ \forall \alpha \ z(\alpha)=1, y(\alpha)=2$ into eq.~\eqref{POVM2} we reproduce their form in the case of the maximally entangled state as a resource state. We get form of the square root measurement which we now is the optimal one in this case
	\be
	\widetilde{\Pi}_i=\sum_{\alpha \vdash N-1}\frac{1}{\sqrt{\rho(\alpha)}}\sigma_a(\alpha)\frac{1}{\sqrt{\rho(\alpha)}},\quad a=1,\ldots,N.
	\ee
	 Making the same substitution in eq.~\eqref{l1a} and using the explicit form of coefficients $c(\alpha,y(\alpha))$ given in Eqn.~\eqref{def_c} and operator $\rho(\alpha)$ we get
	\begin{equation}
	\begin{split}
	F&=\frac{1}{d^{N+2}}\sum_{\alpha \vdash N-1}\sum_{\mu \in \alpha}\lambda_{\mu}(\alpha)^{-1/2}\frac{m_{\mu}}{m_{\alpha}}\tr\left[\sum_{\mu'\in \alpha}\lambda_{\mu'}(\alpha)^{1/2}F_{\mu'}(\alpha) \right]\\
	&=\frac{1}{d^{N+2}}\sum_{\alpha \vdash N-1}\sum_{\mu,\mu'\in \alpha}\lambda_{\mu}(\alpha)^{-1/2}\lambda_{\mu'}(\alpha)^{1/2}d_{\mu'}m_{\mu},
	\end{split}
	\end{equation}
	since $\tr F_{\mu'}(\alpha)=d_{\mu'}m_{\alpha}$. Finally using explicit form of $\lambda_{\mu}(\alpha)=\frac{N}{d^N}\frac{m_{\mu}d_{\alpha}}{m_{\alpha}d_{\mu}}$ we have
	\be
	\label{F_max}
	F=\frac{1}{d^{N+2}}\sum_{\alpha \vdash N-1}\sum_{\mu',\mu \in \alpha} \sqrt{d_{\mu}m_{\mu}}\sqrt{d_{\mu'}m_{\mu'}}=\frac{1}{d^{N+2}}\sum_{\alpha \vdash N-1}\left(\sum_{\mu \in \alpha} \sqrt{d_{\mu}m_{\mu}} \right)^2.
	\ee
	We reproduce the formula for the fidelity of the teleported state from~\cite{Stu2017}.
\end{proof}  
We can also reproduce expression for the fidelity of the teleported state in the case of the maximally entangled state using certain choice of the coefficients $p_{\mu}(\alpha)$ in the most general form of the POVM given by~\eqref{assum2}. 
\begin{corollary}
	Choosing coefficients $p_{\mu}(\alpha)$ in the decomposition~\eqref{assum2} as
	\be
	\forall \alpha  \ \forall \mu\in\alpha \quad p_{\mu}(\alpha)=\frac{1}{\sqrt{\lambda_{\mu}(\alpha)}}=\sqrt{\frac{d^N }{N}\frac{m_{\alpha}d_{\mu}}{d_{\alpha}m_{\mu}}},
	\ee
	and plugging them in~\eqref{cos1} we reproduce fidelity for the maximally entangled state as a resource state (see Theorem 12 of~\cite{Stu2017} or expression~\eqref{F_max} above).
\end{corollary}

\subsection{Convergent algorithm for computing fidelity}
\label{algo}
We now describe a method of approximation of maximal eigenvalues
and corresponding eigenvector of principal submatrices $M_{F}^{d}$~\footnote{SageMath code for implementing the algorithm as well as routines to generate the respective matrices is available upon request.}. We use this algorithm for $2<d<N$, since in this regime we do not know an analytical expressions for maximal eigenvalue and corresponding eigenvector of matrix $M_F$ which are required for computation of $F_{opt}$ together with optimal state and POVM.
From Fact~\ref{M_F_irr} and Corollary~\ref{c_prim} from Section~\ref{central} we can apply Frobenius-Perron theorem (see Theorem~\ref{Perron} of Appendix~\ref{C}) to $M_F$ as well as to all of its principal submatrices $M_F^d$, and write 
\begin{proposition}
If matrix $A\in \M(n,\mathbb{R})$ is non-negative and irreducible then it satisfies the following
eigenequation 
\be
Ax=r(A)x, 
\ee
where $x=(x_{i}):\sum_{i}x_{i}=1$ and $x_{i}>0,$ so this eigenvector is
positive. Such a vector $x$ is called Perron eigenvector of the matrix $A$.
\end{proposition}

Making use of irreducibility and the primitivity, one can approximate maximum eigenvalues and find the corresponding eigenvector of $M_F^d$, which are positive semi-definite and primitive (see Corollary~\ref{c8} and Remark~\ref{posMd}). 

\begin{theorem}
Let $A\in \M(n,\mathbb{R})$  be a positive semi-definite and primitive matrix (in particular $M_F^d$). Suppose that the
vector $w^{0}$ is of the form 
\be
w^{0}=(w_{i}^{0}):\sum_{i}^{n}w_{i}^{0}=1,\quad w_{i}^{0}>0,
\ee
then we define
\be
v^{m+1}=Aw^{m},\quad m=0,1,\ldots \qquad w^{m+1}=\frac{v^{m+1}}{%
\sum_{j}v_{j}^{m+1}},\quad m=0,1,\ldots
\ee
We thus have the following limits 
\be
\lim_{m\rightarrow \infty }w^{m}=x,\qquad \lim_{m\rightarrow \infty
}\sum_{j}^{n}v_{j}^{m}=r(M),
\ee
where $x$ is Perron eigenvector of the matrix $A$. So the sequence of
vectors $\{$ $w^{m}\}$ approximates Perron eigenvector of the matrix $A$,
whereas the number sequence $\{\sum_{j}^{n}v_{j}^{m}\}$ approximates the
spectral radius $r(A)$ of the matrix $A$.
\end{theorem}

\begin{proof}
In the proof we use the  method of calculation of eigenvalues of
diagonalisable matrices described in~\cite{Lancaster}, and for sake of completeness of this manuscript we adopt this method to
our particular case of positive semi-definite, non-negative and irreducible matrices. 

By induction using the non-negativity and irreducibility of the matrix $A$ we get
\be
\forall m\in \mathbb{N} \quad v^{m}=(v_{i}^{m}):v_{i}^{m}>0\Rightarrow \sum_{j}v_{j}^{m}>0,
\ee
so the vectors $w^{m}$ are well defined. From our assumptions on  the
matrix $A$ and Perron-Frobenius Theorem it follows that $A$ has the
following spectral decomposition%
\be
A=\sum_{k=1}^{K}\mu _{k}P_{k},
\ee
where $\mu _{1}=r(A)>\mu _{t}:t\geq 2$ and  $P_{1}=p_{1}p_{1}^{\dagger}:p_{1}=
\frac{w}{||w||}\in \mathbb{R}^{n}$. The vector $w$ is the Perron vector of the matrix $A$, so it
satisfies $w=(w_{i}):\sum_{i}w_{i}=1,$ $w_{i}>0$. The remaining projectors
have the form the standard form 
\be
P_{k}=\sum_{l}p_{k}^{l}p_{k}^{l \dagger}:p_{k}^{l}=(p_{ki}^{l})\in 
\mathbb{R}^{n},\quad ||p_{k}^{l}||=1,\quad k\geq 2.
\ee
Using this spectral decomposition we calculate 
\be
v^{1}=(v_{i}^{1})=\mu _{1}p_{1}(p_{1},w^{0})+\sum_{k\geq 2}\mu
_{k}\sum_{l}p_{k}^{l}(p_{k}^{l},w^{0}),
\ee
where $(p_{1}, w^{0})$ is the standard, Euclidean scalar product of vectors
in the space $\mathbb{R}^{n}$. From this we get
\be
\sum_{j}v_{j}^{1}=\mu _{1}s(p_{1})(p_{1},w^{0})+\sum_{k\geq 2}\mu
_{k}\sum_{l}s(p_{k}^{l})(p_{k}^{l},w^{0}),
\ee
where $s(x)=\sum_{i=1}x_{i}$ for $x=(x_{i})\in \mathbb{R}^{n}$. So we have 
\be
w^{1}=\frac{\mu _{1}p_{1}(p_{1},w^{0})+\sum_{k\geq 2}\mu
	_{k}\sum_{l}p_{k}^{l}(p_{k}^{l},w^{0})}{\mu
	_{1}s(p_{1})(p_{1},w^{0})+\sum_{k\geq 2}\mu
	_{k}\sum_{l}s(p_{k}^{l})(p_{k}^{l},w^{0})}.
\ee
By induction we get
\be
w^{m}=\frac{\mu _{1}^{m}p_{1}(p_{1},w^{0})+\sum_{k\geq 2}\mu
	_{k}^{m}\sum_{l}p_{k}^{l}(p_{k}^{l},w^{0})}{\mu
	_{1}^{m}s(p_{1})(p_{1},w^{0})+\sum_{k\geq 2}\mu
	_{k}^{m}\sum_{l}s(p_{k}^{l})(p_{k}^{l},w^{0})}
\ee
and 
\be
\sum_{j=1}^{n}v_{j}^{m+1}=\frac{\mu
	_{1}^{m+1}s(p_{1})(p_{1},w^{0})+\sum_{k\geq 2}\mu
	_{k}^{m+1}\sum_{l}s(p_{k}^{l})(p_{k}^{l},w^{0})}{\mu
	_{1}^{m}s(p_{1})(p_{1},w^{0})+\sum_{k\geq 2}\mu
	_{k}^{m}\sum_{l}s(p_{k}^{l})(p_{k}^{l},w^{0})},
\ee
where $\mu _{1}=r(A)>\mu _{t} \ : \ t\geq 2$ and $s(p_{1})=\frac{\sum_{i}w_{i}%
}{||w||}=\frac{1}{||w||}>0$, $(p_{1},w^{0})>0$. We thus have
\be
\lim_{m\rightarrow \infty }w^{m}=\frac{p_{1}}{s(p_{1})}=w,\qquad
\lim_{m\rightarrow \infty }\sum_{j=1}^{n}v_{j}^{m+1}=\mu _{1}=r(M).
\ee
\end{proof}

\section{Conclusions and discussion}

We showed that the question of optimal functioning of the dPBT can be reduced to finding a maximal eigenvalue of a certain class of matrices which encode the relationship between Young diagrams. Remarkably, this teleportation protocol can be fully characterized in terms of a single `static' object -- the Teleportation Matrix. This brings about a question on whether one could reduce the study of the optimal performance of other important LOCC protocols in Quantum Information Processing to a study of a simple object which encodes the relationship between the given input and the desired output states of such a protocol analogously do the dPBT.

\section*{Acknowledgements}
MS~is supported by the grant "Mobilno{\'s}{\'c} Plus IV", 1271/MOB/IV/2015/0 from the Polish Ministry of Science and Higher Education. MH and MM are supported by National Science Centre, Poland, grant OPUS 9.  2015/17/B/ST2/01945.
\appendix
\section{Auxiliary facts and lemmas}
\label{App1}
 The set of vectors $\{|\varphi_{k,r}(\alpha)\>\}_{k=1}^{d_{\alpha}}$ spans the $r$-th irrep of $S(N-1)$ is labelled by Young diagram $\alpha$. Define the following operators
	\be
	E_{kl}^{\alpha}=\sum_{r=1}^{m_{\alpha}}|\varphi_{k,r}(\alpha)\>\<\varphi_{l,r}(\alpha)|,
	\ee
	where $m_{\alpha}$ is a multiplicity of irrep labelled by $\alpha$.
The above operators play an important role in the description of the irreps of the symmetric group, but we skip the details here (see for example Appendix F of~\cite{Stu2017}). 
\begin{fact}
	\label{f1}
Assume, that $F_{\mu}(\alpha)$ are projectors onto irreps of algebra $\mathcal{A}_n^{t_n}(d)$, then
	\be
	\<\varphi_{k,r}(\alpha)|\<\Phi^+|F_{\mu}(\alpha)|\Phi^+|\varphi_{l,s}(\alpha)\>=\frac{1}{d}\frac{m_{\mu}}{m_{\alpha}}\delta_{kl}\delta_{rs},
	\ee
	where vectors $\{|\varphi_{k,r}(\alpha)\>\}_{k=1}^{d_{\alpha}}$ span the $r$-th irrep of $S(N-1)$ labelled by Young diagram $\alpha$. 
\end{fact}

\begin{proof}
	Direct calculation shows that
	\be
	\begin{split}
		&\<\varphi_{k,r}(\alpha)|\<\Phi^+|F_{\mu}(\alpha)|\Phi^+|\varphi_{l,s}(\alpha)\>=\tr\left[|\varphi_{k,r}(\alpha)\>\<\varphi_{l,s}(\alpha)|\ot |\Phi^+\>\<\Phi^+| F_{\mu}(\alpha)\right]\\
		&=\frac{1}{m_{\alpha}}\tr\left[E_{kl}^{\alpha}\ot P_+M_{\alpha}P_{\mu} \right]\delta_{rs}=\frac{1}{dm_{\alpha}}\tr\left[E_{kl}^{\alpha}V^{t_n}(n-1,n)M_{\alpha}P_{\mu}\right]\delta_{rs}=\\
		&=\frac{1}{dm_{\alpha}}\tr\left[P_{\mu}P_{\alpha}V^{t_n}(n-1,n)E_{kl}^{\alpha} \right]\delta_{rs}=\frac{1}{dm_{\alpha}}\tr\left[P_{\mu}P_{\alpha}E_{kl}^{\alpha} \right]\delta_{rs}\\
		&=\delta_{kl}\delta_{rs}\frac{1}{dm_{\alpha}}\tr\left[P_{\mu}E_{ii}^{\alpha}\right]=\delta_{kl}\delta_{rs}\frac{1}{d}\frac{1}{d_{\alpha}m_{\alpha}}\tr\left[P_{\mu}P_{\alpha}\right]=\frac{1}{d}\frac{m_{\mu}}{m_{\alpha}}\delta_{kl}\delta_{rs},   
	\end{split} 
	\ee
	since $\tr\left[P_{\mu}P_{\alpha}\right]=m_{\mu}d_{\alpha}$, and $\tr_n V^{t_n}(N,n)=\mathbf{1}_{N}$.
\end{proof}

\begin{remark}
	\label{def_c}
As a natural consequence of Fact~\ref{f1} we have for $k,l=1,\ldots,d_{\alpha}$ and $r,p=1,\ldots,m_{\alpha}$ the following
\be
\<\Phi_+|\<\varphi_{k,r}(\alpha)|\rho^{-1/y(\alpha)}|\Phi_+\>|\varphi_{l,p}(\alpha)\>= c(\alpha, y(\alpha))\delta_{kl}\delta_{rp},
\ee
where
\be
\label{ccc}
c(\alpha, y(\alpha))\equiv \frac{1}{d}\sum_{\mu \in \alpha}\lambda_{\mu}(\alpha)^{-1/y(\alpha)}\frac{m_{\mu}}{m_{\alpha}},
\ee
and $y(\alpha)$ is an arbitrary non-zero real number depending on Young diagram $\alpha \vdash N-1$.
\end{remark}
Using that $\rho(\alpha)=\sum_{\mu\in\alpha}\lambda_{\mu}(\alpha)F_{\mu}(\alpha)$ we get the desired statement.

\section{Additional facts from general matrix theory}
\label{C}
We begin with a short overview of some basic facts from the matrix theory which are required for the analysis of the spectral properties of the matrix $M_F$ described in Section~\ref{central}. We discuss the notion of irreducibility for the matrices with non-negative entries (which is the case for matrix $M_F$) and primitivity. 

Recall Sylwester's Theorem~\cite{Horn}
\begin{theorem}(Sylwester)
	\label{Sylw}
A Hermitian matrix $A$ is positive semi-definite if and only if all  principal minors are positive.
\end{theorem}

\begin{definition}
\label{irr_mat0}
Let $A\in \M(m,\mathbb{C})$, then the matrix $A$ is irreducible	if it cannot be conjugated into the block upper triangular form by a permutation matrix $P$:
\be
PAP^{-1}\neq \begin{pmatrix}
	A_1 & A_2\\
	0 & A_3
\end{pmatrix},
\ee
where $A_1, A_3$ are non-trivial square matrices.
\end{definition}
If $A\in \M(m,\mathbb{R})$ is non-negative, then we have an equivalent definition (which is the case for the teleportation matrix $M_F$):
\begin{definition}
\label{irr_mat}
Let $A\in \M(m,\mathbb{R})$ be a non-negative matrix, then the matrix $A$ is irreducible if for any pair of indices $1\leq i,j\leq n$ there exists a $q\in \mathbb{N}$ such that $(A^q)_{ij}>0$.
\end{definition}

We now present a stronger version of the Frobenius-Perron theorem:

\begin{theorem}(Frobenius-Perron)
	\label{Perron}
	Let $A$ be an $m\times m$ irreducible matrix with non-negative, real entries with the spectral radius $r(A)$. Then we have the following:
	\begin{enumerate}
		\item  The number $r(A)$ is a positive real number and it is an eigenvalue of matrix $A$ (Perron-Frobenius eigenvalue).
		\item The multiplicity of an eigenvalue $r(A)$ is equal to one.
		\item The matrix $A$ has an eigenvector corresponding to an eigenvalue $r$ with all positive components.
	\end{enumerate}
\end{theorem}

\begin{definition}
	\label{primi}
	A non-negative matrix $A\in \M(m,\mathbb{R})$ is primitive if it is irreducible and  has only
	one non-zero eigenvalue of maximum modulus.
\end{definition}
On the other hand we have~\cite{Horn}:
\begin{proposition}
	\label{primi_prop}
	If the matrix $A\in \M(m,\mathbb{R})$ is non-negative, irreducible, and has positive diagonal then 
	$A$ is primitive.
\end{proposition}
At the end we introduce the notion of centrosymmetric matrices.
\begin{definition}
\label{centro}
Matrix $A\in \M(m,\mathbb{C})$ is called centrosymmetric if its entries satisfy
\be
A_{i,j}=A_{m-i+1,m-j+1}\quad \text{for}\quad 1\leq i,j\leq m.
\ee
\end{definition}
\section{The explicit form of Young projectors and operators $F_{\mu}(\alpha)$ in natural representation}
\label{natural}
We provide the construction of the permutation operators $V(\sigma)$, where $\sigma \in S(N)$, Young projectors $P_{\mu}$, and projectors $F_{\mu}(\alpha)$ in the computational basis. Using this representation we can construct the explicit form of the optimal POVM~\eqref{opt_povm} and state~\eqref{opt_State} for various $N,d$.

Consider a unitary representation of a permutation group $S(N)$ acting on the $N-$fold tensor product of complex spaces $\mathbb{C}^d$, so our full Hilbert space is $\mathcal{H}\cong (\mathbb{C}^d)^{\otimes N}$. For a fixed permutation $\sigma\in S(N)$ a unitary transformation $\opV(\sigma)$ is given by
\be
\label{unitary}
V(\sigma)\left( |e_{i_1}\>\otimes \ldots  \otimes |e_{i_N}\>\right)=|e_{i_{\sigma^{-1}(1)}}\> \otimes  \ldots  \otimes |e_{i_{\sigma^{-1}(N)}}\>,
\ee
where the set $\{|e_{i_1}\>\ot \cdots \ot|e_{i_N}\>\}$ is a standard basis in $(\mathbb{C}^d)^{\otimes N}$. Then, the explicit form of the operator $V(\sigma)$ for some $\sigma \in S(N)$ is given by
\be
\label{swap1}
V(\sigma)=\sum_{e_{i_1},\ldots,e_{i_N}}|e_{i_{\sigma^{-1}(1)}}\> \ot\cdots \ot |e_{i_{\sigma^{-1}(N)}}\>\<e_{i_1}|\ot \cdots \ot \<e_{i_N}|.
\ee
Using an expression for any permutation operator $V(\sigma)$, the explicit form of Young projectors in the natural representation is
\be
\label{exp_young}
P_{\mu}=\frac{f_{\mu}}{N!}\sum_{\sigma \in S(N)}\chi^{\mu}\left(\sigma^{-1} \right)V(\sigma),
\ee
where $\chi^{\mu}(\sigma)$ is the character calculated on the irreducible representation labelled by the Young diagram $\mu \vdash N$ on the permutation $\sigma \in S(N)$, $f^{\mu}$ is some constant depending on the Young diagram $\mu \vdash N$ (see for example~\cite{Fulton1991-book-rep}). The explicit form of the projectors $F_{\mu}(\alpha)$ described briefly in the introductory part of our manuscript (for complete description see~\cite{Stu2017}) are given by
\be
\label{explicit}
F_{\mu}(\alpha)=\frac{1}{\gamma_{\mu}(\alpha)}P_{\mu}\sum_{a=1}^{N}V(a,N)P_{\alpha}\ot \widetilde{P}_+V(a,N)P_{\mu},
\ee
where $P_{\alpha},P_{\mu}$ are Young projectors onto irreducible spaces labelled by Young diagrams $\alpha \vdash N-1$ and $\mu \vdash N-1$ respectively, $\widetilde{P}_+$ is an unnormalised projector onto the maximally entangled state between $N-$th and $n=N+1-$th, and $\gamma_{\mu}(\alpha)$ is given in~\eqref{gamma}.
\bibliographystyle{plain}
\bibliography{biblio}
\end{document}